\documentclass[11pt]{article}
\usepackage{amsfonts,color}
\usepackage{amsmath,amsthm}
\usepackage{hyperref}
\hypersetup{colorlinks,urlcolor=blue}
\usepackage{slashed}
\usepackage[all]{xy}
\usepackage{graphicx}

\usepackage{mathrsfs}
\usepackage{bbm}

\usepackage{amssymb}
\usepackage{latexsym}

\DeclareMathAlphabet{\mathpzc}{OT1}{pzc}{m}{it}
\DeclareMathAlphabet\EuFrak{U}{euf}{m}{n}	%  the Bold Euler Fraktur
\SetMathAlphabet\EuFrak{bold}{U}{euf}{b}{n}	%  gothic font

\newcommand{\End}{ {\bf end} }

\newcommand{\bC} {{\mathbb C}}

\newcommand{\bR} {{\mathbb R}}

\newcommand{\bZ} {{\mathbb Z}}

\newcommand{\eps}{\epsilon}

\newcommand{\mA}{\mathcal A}
\newcommand{\mB}{\mathcal B}

\newcommand{\mH}{\mathcal H}

\newcommand{\mM}{\mathcal M}

\newcommand{\mO}{\mathcal O}
\newcommand{\mP}{\mathcal P}

\newcommand{\mU}{\mathcal U}

\newcommand{\rC}{\mathrm{C}}

\newcommand{\bo}{{\partial_0 b}}
\newcommand{\bl}{{\partial_1 b}}
\newcommand{\bk}{{\partial_k b}}

\newcommand{\ad}{{\mathrm{ad}}}

\newtheorem{thm}{Theorem}[section]

\newtheorem{lem}[thm]{Lemma}
\newtheorem{prop}[thm]{Proposition}

\newtheorem{defn}[thm]{Definition}

\theoremstyle{definition}
\newtheorem{ex}{Example}[section]
\newtheorem{rem}[ex]{Remark}

\theoremstyle{remark}

\numberwithin{equation}{section}

\newcommand{\loc}[1]{{\scriptstyle (#1)}}

\begin{document}

\author{\textsc{Fabio Ciolli${}^1$, Giuseppe Ruzzi${}^2$ and Ezio Vasselli${}^2$} \\
\small{${}^1$Dipartimento di Matematica e Informatica, Universit\`a della Calabria,}\\
\small{Via Pietro Bucci, Cubo 30B, I-87036 Rende, Italy.}  \\
\small{${}^2$Dipartimento di Matematica, Universit\`a di Roma ``Tor Vergata'',}\\
\small{Via della Ricerca Scientifica, I-00133 Roma,  Italy.}  \\
\small{\texttt{fabio.ciolli@unical.it, ruzzi@mat.uniroma2.it, ezio.vasselli@gmail.com }}
}

\title{Where charged sectors are localizable:\\ a viewpoint from covariant cohomology

\bigskip

\large{Dedicated to Roberto Longo on the occasion of his 70th birthday}
}

\maketitle

\begin{abstract} 
Given a Haag-Kastler net on a globally hyperbolic spacetime,
one can consider a family of regions where quantum charges are supposed to be localized.
Assuming that the net fulfils certain minimal properties 
(factoriality of the global observable algebra and relative Haag duality),
we give a geometric criterion that the given family must fulfil 
to have a superselection structure with charges localized on its regions.
Our criterion is fulfilled by all the families used in the theory of sectors
(double cones, spacelike cones, diamonds, hypercones).
In order to take account of eventual spacetime symmetries,
our superselection structures are constructed in terms of covariant charge transporters, 
a novel cohomological approach generalizing that introduced by J.E. Roberts. 
In the case of hypercones, with the forward light cone as an ambient spacetime, 
we obtain a superselection structure with Bose-Fermi parastatistics and
particle-antiparticle conjugation. 
It could constitute a candidate for a different description
of the charged sectors introduced by Buchholz and Roberts for theories including massless particles.

\medskip

\noindent 
{\bf Mathematics Subject Classification.} 81T05, 81T20  \\
{\bf Keywords.} Algebraic quantum field theory, superselection sectors.

\end{abstract}

%\tableofcontents

\section{Introduction}
The analysis of superselection sectors has been a topic of interest
since the early days of algebraic quantum field theory, 
in particular in the celebrated DHR analysis \cite{DHR69a,DHR69b,DHR71,DHR74}. 
There, superselection sectors are intrinsically characterized
in terms of endomorphisms of the observable C*-algebra $\mA$,
realized on the vacuum Hilbert space $\mH$. 
The endomorphisms of interest, say $\rho \in \End\mA$, are localized in 
double cones $\mO \subset \bR^4$, in the sense that
$\rho$ equals the identity 
on any subalgebra of $\mA$ generated by observables localized
in regions causally disjoint from $\mO$ (\emph{DHR endomorphisms}).
The relation with superselection sectors in the sense of Wigner, Wick and Wightman, 
that define Hilbert space representations of the type $\pi : \mA \to \mB(\mH_\pi)$, is that there is a 
unitary equivalence
$\pi \simeq \rho$
having regarded $\rho$ as a representation $\rho : \mA \to \mA \subset \mB(\mH)$.
This yields the DHR criterion to select the endomorphisms of interest:
the physical interpretation is that there is a quantum charge localized in $\mO$
that induces a superselection sector disjoint from the vacuum. 

\medskip 

An alternative way to look at DHR superselection sectors  was proposed by J.E. Roberts. Guided by the concept of parallel transport, Roberts introduced a families  of unitary operators $X_b \in \mA$,  
where $b$ are triples of double cones $\bo , \bl \subseteq |b| \subset \bR^4$,  \cite{Rob90,Rob04}. These unitaries  satisfy a cocycle relation and  realize  an equivalence between the DHR endomorphisms
$\rho\loc{\bl}$ and  $\rho\loc{\bo}$ localized in $\bl$ and $\bo$ respectively, as
\begin{equation}
\label{eq.intro.02}
X_b \, \rho\loc{\bl}(A) \ = \ \rho\loc{\bo}(A) \, X_b 
\ \ \ , \ \ \ 
A \in \mA 
\, .
\end{equation}
Hence, the physical idea is that a quantum charge, localized in $\bl$, is transported to $\bo$
along a curve $\vec{b} : [0,1] \to \bR^4$
with $\vec{b}([0,1]) \subset |b|$, $\vec{b}(0) \in \bl$, $\vec{b}(1) \in \bo$. 
The DHR endomorphisms appearing in (\ref{eq.intro.02}) are
constructed by adjoint action of $X_b$, by letting $\bo$ go to spacelike infinity.

\medskip 

In Minkowski space the Roberts approach is equivalent to the classical DHR formulation. 
Yet in curved spacetimes it turned out that charge transporters are unavoidable, 
because the DHR endomorphisms alone do not carry all the necessary physical informations \cite{Ruz05,BR09}. 
Successive work has shown that the analogy with parallel transport is fully
justified by the fact that, when $\mA$ is the observable net of the free Dirac field $\psi$,
for any charge transporter $X$ there is a closed 1--form $A$ on $\mM$, such that
\begin{equation}
\label{eq.intro.01}
X_b \ = \ \psi_\bo^* \, e^{i \int_{\vec{b}} A } \, \psi_\bl 
\ \ \ , \ \ \ 
\rho_{(\bk)} \, = \, \ad \psi_\bk 
\ \ \ , \ \ \  
k =0,1
\, .
\end{equation}
Here, $\psi_\bo$, $\psi_\bl$ are charged unitaries obtained by smeared field operators
by using the III-type factor property of the local algebras, inducing the DHR endomorphisms
by adjoint action. They carry opposite charges, coherently with the ideas of \cite{HK64,BF82}.
The physical interpretation is that $A$ is an external potential interacting with $\psi$,
inducing observable effects in spacetimes with non-trivial fundamental group 
(\emph{Aharonov-Bohm effect}) and vanishing in simply connected spacetimes.
Expressions analogous to (\ref{eq.intro.01}) can be obtained for a non-Abelian gauge group, 
with $A$ a flat connection with values in the corresponding Lie algebra \cite{Vas15,Vas19,DRV20}
\footnote{Charge transporters in the universal C*-algebra of the quantum electromagnetic field \cite{BCRV16,BCRV17} similar to \eqref{eq.intro.01} have been constructed in \cite{BCRV19,BCRV22,BCRV23}. We quote also \cite{CRV12, CRV15} where connection operators defined on paths are associated with representations for the quantum electromagnetic field.}. 

\medskip

A different superselection structure was proposed by Buchholz and Fredenhagen, which introduced a criterion
based on positivity of the energy and arrived to a charge localization in spacelike cones \cite{BF82}. 
This type of localization reflects the idea that a spacelike cone
is able to host two opposite charges, one of which is free to go to (or to come from) spacelike infinity,
\emph{behind the moon} following Haag and Kastler \cite{HK64}. 
This introduces a further level for the regions of interest: 
we have bounded regions (double cones) in which the observables are localized, 
and unbounded regions (spacelike cones) encoding the localization of charges.
This suggests to define by additivity a new observable net, say $\mA_K$, 
indexed by the charge localization regions instead of the usual family of double cones. 
As spacelike cones are not upward directed under inclusion, the so-obtained sectors 
are now described by net morphisms $\rho : \mA_K \to \mB(\mH)$ (\emph{BF-morphisms})
that in general do not match to form endomorphisms of the global C*-algebra. 
In spite of this difficulty, 
they form a mathematical structure analogous to the one obtained by DHR sectors, namely a symmetric tensor category with conjugates
encoding a permutation symmetry with Bose-Fermi parastatistics and particle-antiparticle correspondence.

\medskip 

In lower dimensions, because of topological obstructions different properties manifest themselves at the level
of statistics. The braid group appears instead of the permutation group when the exchange of charges is performed,
and exotic statistics appear. 
Such sectors have been studied in \cite{FRS89, FRS92}. 
In the present paper we focus on four spacetimes dimensions, so that we will always deal with Bose-Fermi (para)statistics.

\medskip 

All the previous approaches share an assumption, namely the absence of massless particles.
This of course excludes QED, which presents peculiarities that create serious problems
for the identification of superselection sectors.
Notably, the electric charge cannot be localized (not even in a spacelike cone), 
and charged fields cannot be local to the electromagnetic field 
in positive (and hence physical) gauges.
Both the problems have their root in the Gauss law \cite[\S 7.2]{Strocchi}.
The search for a superselection structure describing the electric charge
was approached in several papers \cite{FMS79A,Buc82,BDMRS96,BDMRS07},
culminating in the analysis by Buchholz and Roberts \cite{BR14}.
There, a crucial point is that the ambient spacetime is restricted 
to the future light cone $V_+$, in such a way to rule out the infrared clouds
giving rise to the overabundance of sectors with the same electric charge.
The restriction on $V_+$ has also consequences as the loss (then amended) of the full group of Poincar\'e symmetry,
which is restricted to an action by the semigroup stabilizing $V_+$,
and the eventuality that $\mA(V_+)$ (now generated by observables localized in $V_+$)
is not irreducible. 
This has an impact on Haag duality, that now must be referred to the von Neumann algebra $\mA(V_+)''$
that is not necessarily $\mB(\mH)$ but more generally a ${\mathrm{III}}_1$ factor.
The result is a family of approximately inner morphisms $\sigma : \mA(V_+) \to \mA(V_+)''$,
localized on hypercones and giving rise to the desired superselection structure (\emph{BR-sectors})
even if some questions remain open for parastatistics.

\

The present paper has two main objectives. 

\textit{First}, identify the essential features that define a class of regions within 
a globally hyperbolic spacetime which can serve as localization regions for charged superselection sectors. To emphasize this point, we describe these properties using a partially ordered set (poset), and consider the observable net defined on it.

\textit{Secondly}, develop a novel approach to superselection sectors 
in terms of \emph{covariant 1-cocycles}: charge transporters that maintain covariance with respect to a global symmetry of the spacetime. 
This 
concept extends the cohomology introduced by Roberts and addresses a gap in this approach, 
as the covariant properties of superselection sectors have previously only been discussed in terms of localized and transportable morphisms.
This aspect holds significance because it allows, in principle, for the examination of covariant superselection sectors 
on spacetimes with nontrivial topology. 
In such cases, in fact, the conventional approach relying on localized and transportable morphisms may prove inadequate \cite{BR09}.
  
We demonstrate that our superselection structure possesses permutation symmetry with Bose-Fermi (para)statistics, and a conjugation that encodes the particle-antiparticle correspondence. This machinery applies to the known situations described in the previous lines. Yet, whilst for double cones and spacelike cones we easily 
get the DHR and BF superselection structures respectively, 
in the case of hypercones it remains an open question whether our superselection structure 
agrees with the one constructed by Buchholz and Roberts. 
We postpone this discussion to future works, with the remark that, due to the generality of our approach, we did not take into account the energetic aspects
that come into play when the transport of a charge towards the boundary of the light cone is performed
{\footnote{
The authors are indebted with Detlev Buchholz for discussions on this point.
}}.

\ 

In the following lines we give a sketch of the content of the next sections. 

\medskip 

In \S \ref{obs} we fix a poset modelled on a family of regions 
of a spacetime $\mM$ encoding the charge localization and the geometry of the spacetime.
Namely, we consider a possibly curved spacetime denoted as $\mM$, foliated along time
and equipped with a symmetry action given by spacetime transformations.
In general, only symmetries that preserve the charge localization regions are of interest, 
and in all the mentioned examples they form a group
(even in the case of charges of electromagnetic type, 
despite the stabilizer of the forward light cone is a semigroup of Poincar\'e transformations,
the symmetries that stabilize the hypercones are given by the Lorentz group).

Hence, we consider an abstract  poset $K$ endowed with a causal disjointness relation $\perp$ and acted upon by a symmetry $\mP$, and determine two properties 
that the elements of $K$ must fulfill to play the role of localization regions 
in a superselection structure expressed in terms of charge transporters. 
The two properties are that the causal complement of any $o \in K$ must be pathwise connected, 
and that
for any pair $o_1,o_2$, if  $o_1\perp o_2$ then there is 
$a\perp o_1,o_2$, otherwise there are  $\tilde o_1,\tilde a$ such that $\tilde o_1\subset o_1$ and 
$\tilde a\perp \tilde o_1,o_2$. 
The former property ensure that we can freely transport charges along paths
in the causal complement of $o$, whilst the latter allows to add further charges
without interfering with a previously given set of mutually casually disjoint charges.

\medskip 

We then consider in \S \ref{rep} local von Neumann algebras $\mA(o)$ 
for any element of $o$ of $K$ defined in a reference Hilbert space $\mH$
\footnote{  
In the examples mentioned above, spacelike cones for instance, 
$\mA(o)$ is generated by the von Neumann algebras of observables
localized in double cones contained in $o$ and defined in the vacuum Hilbert space.}. 
The correspondence 
$\mA_K: o\mapsto \mA(o)$ gives a net of von Neumann algebras  
over $K$ on which we require minimal assumptions derived from the existing theory of superselection sectors: the 
symmetry action of $\mP$ on $\mA_K$ is unitarily implemented; the $\rC^*$-algebra
$\mA(K)$ generated by  $\mA(o)$ as $o$ varies in $K$ is factorial,
$\mA(K)'\cap \mA(K)''=\mathbb{C}\, 1_{\mH}$; 
relative Haag duality holds \cite{Cam07},   
\begin{equation}
\label{eq.intro.04}
 \mA(o^\perp)' \cap \mA(K)'' =  \mA(o) \ , \qquad o \in K \, ,
\end{equation}
where $\mA(o^\perp)$ is the von Neumann algebra associated with the causal complement
$o^\perp$ of $o$.

In \S \ref{cocycles}, starting from $\mA_K$ we introduce the notion of \emph{covariant cocycle}. 
It simultaneously yields a cocycle (charge transporter) in the sense of Roberts and a cocycle over $\mP$.
Covariant cocycles are designed to encode in a concise way $\mP$-covariant superselection sectors.
Considering for simplicity only simply connected spacetimes (hence simply connected posets) 
in order to rule out the above mentioned Aharonov-Bohm vector potentials,
we prove that covariant cocycles form a C*-category $Z^1_c(\mA_K)$.
A further analysis shows that $Z^1_c(\mA_K)$ has permutation symmetry with Bose-Fermi (para)statistics and conjugates, 
and hence is endowed with the properties
characterising a charged superselection structure.

$\mM$ being simply connected, $Z^1_c(\mA_K)$ can be described in terms of morphisms 
from $\mA_K$ into $\mA(K)$ in the sense of \cite{RV11};
these morphisms are localized on regions $o \in K$ \S \ref{c}, 
and in good cases (e.g. when $K$ is the set of double cones), 
they form actual endomorphisms of $\mA(K)$.
For various choices of $\mM$ and $K$, we get
DHR sectors (when $\mM = \bR^4$ and $K$ is given by double cones), 
BF sectors ($\mM = \bR^4$, $K$ given by spacelike cones),
and sectors in curved spacetimes ($\mM$ a globally hyperbolic spacetime, $K$ the family of diamonds).
When $\mM$ is the light cone and $K$ is formed by hypercones, we get a superselection structure
describing hypercone localized charges as explained in the previous lines.

\medskip 

In the final \S \ref{sec.concl} we expose some conclusions and discuss future developments of our work.

\section{The observable net}
\label{obs}
The investigation of charged superselection sectors within a given spacetime requires to consider observables in regions 
where the charges are localized, on the ground of a physical criterion. 
This leads to a correspondence assigning to such charge localization regions the algebras of observables measurable within them.
As usual, such algebras are obtained by additivity, starting from Haag-Kastler nets.
In the following we shall refer to the set of charge localization regions as the \emph{set of indices}. \smallskip

\subsection{The set of indices}
\label{Kn}

In the present section we discuss the essential properties defining a set of indices for due purposes. 
For convenience we regard it as an abstract partially ordered set (\emph{poset}) 
whose properties are modelled on a suitable family of open regions, ordered under inclusion, 
of a connected and simply connected 4-dimensional globally hyperbolic spacetime $\mM$ having non-compact Cauchy surfaces, 
and endowed with an action by spacetime symmetries when present. 
As illustrated in the introduction, the hypothesis of $\mM$ being
simply connected is motivated by the wish of simplifying the exposition, excluding Aharonov-Bohm type effects. 
Anyway we remark that, even if $\mM$ is simply connected, $K$ may be not directed. 
\smallskip

Let $K$ denote a non-empty partially ordered set with order relation $\subseteq$.
Elements of $K$ are denoted by Latin letters $o,a...$.  
In the following lines we give the list of properties that we assume for $K$: with the exception of 
\textbf{K4} and \textbf{K6} 
that are the crucial ones, they are natural requirements dictated by the idea that $K$ is the abstraction of a 
family of causally complete regions in a 4-dimensional spacetime, stable under the action of a group $\mP$ 
that is either the full group of spacetime symmetries or a subgroup thereof.
\begin{itemize}
\item[\textbf{K1}] for any $o\in K$ there are $\hat o, \tilde o\in K$ with $\hat o\subset o\subset \tilde o$.
\end{itemize}
In the previous line, the symbol $\subset$ stands for any relation implying $o \subseteq \tilde o$ and $o \neq \tilde o$,
and shall be defined concretely when we will consider examples of $K$.
%
%{\marginpar{verificare $\subset$}}
%
The poset  $K$ is endowed  with a \emph{causal disjointness relation} i.e. a symmetric, non-reflexive binary relation $\perp$,
compatible with the order relation in the sense that
\begin{itemize}
 \item[\textbf{K2}] 
 $o\perp a \ , \ \hat o\subseteq o \ \Rightarrow \ \hat o\perp a$ 
\end{itemize}
The \emph{causal complement} of $o$ is the subset $o^\perp$ of $K$ defined as 
$o^\perp :=\{a\in K \ | \ a\perp o \}$. We assume that all the elements of $K$ are \emph{causally complete}
in the sense that
\begin{itemize}
\item[\textbf{K3}] 
$o^{\perp\perp} = 
\{\tilde o \in K \ | \ \tilde o\perp a \ , \ \forall a\in o^\perp\} = 
\{ a\in K \ | \ a \subseteq o \}$ 
\end{itemize}
Note that this implies that $o^\perp\ne\emptyset$. 
A key property is the following 
\begin{itemize}
\item [\textbf{K4}]
for any pair $o,a\in K$, if $o \perp a$ then $o^\perp\cap a^\perp\ne \emptyset$ otherwise 
there is $\tilde o\subset o$ s.t. ${\tilde o}^\perp\cap a^\perp\ne \emptyset$
\end{itemize}
We assume that $K$ is acted upon by a group $\mP$, \emph{the symmetry}, 
and that this action is compatible with the order and causal disjointness relations
\begin{itemize}
\item[\textbf{K5}] $a\subset o \ \Rightarrow \ \lambda a\subset \lambda o \ \ \text{ and } \ \ a\perp o \ \Rightarrow \ \lambda a\perp \lambda o$ 
\end{itemize}
Before to proceed with our list,
for the reader's convenience we introduce some properties of $K$  having a topological flavour, 
expressed in terms of the set of \emph{singular $n$-simplices}  $\Sigma_n(K)$ 
(see \cite{Ruz05} for details).
A $0$-simplex is just an element of $K$, that we usually write $a,o$.
A $1$-simplex is given by a triple $b$ of $0$-simplices $\partial_0b , \partial_1b \subseteq |b|$;
we call $|b|$ the \emph{support} and $\partial_0b , \partial_1b$ the \emph{faces}.
Given the 1-simplex $b$, the \emph{reverse} 
$\overline{b}$ is the 1-simplex having the same support as $b$ 
and such that $\partial_0\overline{b}=\partial_1b$,
$\partial_1\overline{b}=\partial_0b$. A 1-simplex $b$ is said to be \emph{degenerate} 
whenever $\partial_1b = \partial_0b$. In particular, we denote the degenerate 1-simplex
whose faces and support equal $o\in\Sigma_0(K)$ by $\sigma_0(o)$.
Inductively, for $n\geq 1$, an $n$-simplex $x$ is formed by $n+1$, $(n-1)$-simplices
$\partial_0x,\ldots, \partial_nx$ and by an element of the poset $|x|$, 
the \emph{support} of $x$, such that $|\partial_ix|\subseteq |x|$ for $i=0,\ldots,n$;
on the other hand, any $(n-1)$-simplex $x$ can be regarded as a degenerate $n$-simplex
$\sigma_i(x)$ obtained by repeating its $i$ vertex.
The symbols $\partial_i$ and $\sigma_i$ define face and degeneracy maps respectively. 
To be concise, we write the compositions 
$\partial_{ij} \doteq \partial_i\partial_j$, 
$\sigma_{ij} \doteq \sigma_i\sigma_j$. 
It can be proved that
\[
\partial_{ij} = \partial_{j(i+1)} \ , \qquad i\geq j
\]
A \emph{path} $p:a\to o$ is a finite  ordered set of 1-simplices $b_n*\cdots *b_1$
satisfying the relations $\partial_0b_{i-1}= \partial_1b_{i}$ 
for $i=2,\ldots,n$ and $\partial_1b_1=a, \partial_0b_n=o$.
%We shall denote $\partial_1p=a$ and $\partial_0p=o$, the \emph{starting and the ending point} 
%of a path $p$.  
The \emph{support} $|p|$ of a path $p$ is the collection of the supports of the 1-simplices by which it is composed.
The \emph{reverse} of $p$ 
is the path $\overline{p}:o\to a$ defined by 
$\overline{p}\doteq \overline{b}_1*\cdots *\overline{b}_n$. 
If $q$ is a path from $o$ to $\hat
o$, then we can define, in an obvious way, the composition 
$q*p:a\to \hat o$. A subset $S \subseteq K$ is said to be \emph{pathwise connected}
whenever for any pair $a$,$\tilde a$ of 0-simplices in $S$
there is a path from $a$ to $\tilde a$.
We are now ready to give our last properties.
\begin{itemize}
\item[\textbf{K6}] We assume that the $K$ is pathwise connected and that the causal complement $o^\perp$ 
of any element $o\in K$ is pathwise  connected.
\end{itemize}
An \emph{elementary deformation} 
of a path $p=b_n*\cdots*b_1$ consists in replacing a 1-simplex $\partial_1c$ of the path by
the pair $\partial_0c*\partial_2c$, where $c\in\Sigma_2(K)$, 
or conversely in replacing a consecutive pair
$\partial_0c*\partial_2c$ by a single 1-simplex $\partial_1c$.
Two paths with the same endpoints are \emph{homotopic} 
if they can be obtained from one another by a finite sequence of
elementary deformations. Homotopy defines an equivalence relation $\sim$ 
on the set of paths with the same endpoints and, $K$ being pathwise connected,  
this naturally leads to the notion of the fundamental group $\pi_1(K)$ of $K$. 
In all the cases of interest in the present paper, it turns out $\pi_1(\mM) \simeq \pi_1(K)$ \cite{Ruz05}.
The following hypothesis is formulated to simplify the discussion of the next section:
\begin{itemize}
\item[\textbf{K7}] We assume that $K$ is simply connected, i.e.\, $\pi_1(K)=1$. 
\end{itemize}
% \red{Tolto "the"}
%
%
%
We finally observe that the action of the symmetry $\mP$ on $K$ extends to simplices by taking 
$\lambda x$ as the n-simplex whose faces are $\lambda\partial_i x$ and whose support is 
$|\lambda x|=\lambda|x|$. This induces an action on paths, which preserves the homotopy equivalence relation, defined by  
\[
\lambda p:= \lambda b_n*\cdots *\lambda b_1 \ , \qquad 
p=b_n*\cdots *b_1 \ .  
\]
Note that if $p:a\to o$ then $\lambda p:\lambda a\to \lambda o$.  

\subsection{Observable net and reference representation}
\label{rep}
We now add the further ingredients given by the observable net and its reference representation. 
Since we are not assuming that $K$ is upward directed, we have to take into account 
the notions of morphism and representation of a net, as in \cite{RV11}. 
%\marginpar{Rimosso il riferimento all'appendice}

\medskip 

Let $\mH$ be a Hilbert space. An \emph{observable net over $K$} is given by a correspondence 
\[
\mA_K:K\ni o \mapsto \mA(o)\subseteq \mB(\mH) \, ,
\]
assigning to $o \in K$ the von Neumann algebra $\mA(o)$.
We denote the $\mathrm{C}^*$-algebra generated by $\cup_{o\in K} \mA(o)$ by $\mA(K)$,
and assume the standard properties of isotony, causality,
and $\mP$-covariance
under a representation $U : \mP \to \mU(\mH)$. 
Finally we assume 
\begin{itemize}
\item \emph{factoriality},
\[
\mA(K)' \cap \mA(K)'' \, = \, \mathbb{C}1_{\mH} \, ;
\]
\item the \emph{Borchers property}: for any $o\subset a$ and for any projection $E\in\mA(o)$ there is an isometry $V\in\mA(a)$  such that  $VV^*=E$;
\item \emph{relative Haag duality}, 
\[
\mA(o) \, = \, \mA(K)'' \cap \mA(o^\perp)'
\ \ \ , \ \ \ 
o \in K
\, ,
\]
\end{itemize}
where, in our abstract setting, we may take as a definition $\mA(o^\perp)' = \bigcap_{a\perp o} \mA(a)^{\prime}$.
Note that factoriality is a weaker assumption with respect to the irreducibility property usually required. It is dictated by the hypothesis made by Buchholz and Roberts for their net defined on the light cone \cite{BR14}. 
In this connection, note that when the net is irreducible relative Haag duality becomes the usual Haag duality.

\subsection{Meaningful examples}

We now give some concrete examples where a set of indices $K$ and observables net over $K$ satisfying the properties introduced above arise.
The general context is that of a connected and simply connected 4-dimensional globally hyperbolic spacetime $\mM$,
with non-compact Cauchy surfaces and, possibly, symmetries (i.e. global isometries). 
$K$ can be any suitable family of open subsets of $\mM$, ordered under inclusion and with the causal disjointness relation induced by the causal structure of $\mM$.
To be precise:
\begin{itemize}
\item the proper inclusion $o\subset a$ amounts to the inclusion of the closure $cl(o)$ of $o$ in $a$;
\item the causal disjontness relation $o\perp a$ amounts to $J(cl(o))\cap cl(a)=\emptyset$ where $J(cl(o))$ denotes the causal set of $cl(o)$;
\item the action $\lambda o$ of a symmetry $\lambda$ on a region $o$ is nothing but the image $\lambda(o)$;
\item finally, the symmetry $\mP$ is the subgroup of symmetries of $\mM$ that leave invariant $K$.
\end{itemize} 
The defining properties of $K$ and the assumptions we made on $\mA_K$ are verified in the following cases.
\begin{enumerate}
\item \emph{The DHR sectors} \cite{DHR71,DHR74}. Here $\mM$ is the 4-dimensional Minkowski space, $\mP$ the Poincar\'e group, $K$ the set of double cones and $\mA_K$ is the observable net defined in the vacuum representation.

\item \emph{The DHR sectors in curved spacetimes} \cite{GLRV01,Ruz05,BR09}. Here $\mM$ is an arbitrary 4-dimensional globally hyperbolic spacetime with non-compact Cauchy surfaces, 
$K$ is the set of diamonds and  $\mA_K$ is  the observable net in a representation satisfying the microlocal spectrum condition.
\item \emph{The BF sectors} \cite{BF82}. Here $\mM$ the 4-dimensional Minkowski spacetime,  
$K$ the set of spacelike cones, $\mP$ the Poincar\'e group, and $\mA_K$ is defined in the vacuum representation. It is worth observing 
that being spacelike cones unbounded regions of $\mM$, one first consider the observable net over double cones defined in the vacuum representation. Then the von Neumann algebra associated with a spacelike 
cone $o$ is defined as the von Neumann algebra generated by the algebras  $\mA(\mO)$ associated with double cones $\mO$ such that $\mO\subset o$.
\item \emph{The BR charge classes (charges of electromagnetic type)} \cite{BR14}. Here $\mM$ is the forward light cone in the 4-dimensional Minkowski space and $K$ is the set of hypercones in $\mM$. Although $\mM$ has a semigroup of symmetries given by the semidirect product of non-negative time translations and the Lorentz group, $K$ is not stable under non-negative time translations. As a consequence $\mP$ corresponds to the Lorentz group. $\mA_K$ is defined by the observable net in the vacuum representation: since hypercones are unbounded, elements of $\mA_K$ are generated, similarly to the BF sectors, by the algebras of double cones that are located within each hypercone.
\end{enumerate}
We conclude by noting that violations of the properties \textbf{K1-7}, to the authors' knowledge, occur in low-dimensional spaces. 
For example, in 2-dimensional Minkowski space the causal complement of double cones is not connected, see \cite{FRS89}. 
In 3-dimensional Minkowski space, spacelike cones have connected complement but form a poset with non-trivial homotopy group, 
isomorphic to $\bZ$, ref.\ \cite{Mun09, Naa11}.
These violations and other topological restrictions like for instance that in these posets there are pairs of causally disjoint elements whose causal complements have non-connected intersection are the reasons for the appearance of 
anyonic sectors.

\section{Covariant cohomology and charged representations of the observable net}
\label{cocycles}

In the present section we introduce covariant 1-cocycles and, following the DHR analysis, 
we show that they define a symmetric tensor $\rC^*$-category
with conjugates, closed under direct sums and subobjects. 
This implies that charged quantum numbers are associated with the equivalence class of any irreducible covariant 1-cocycle. 
Using the hypothesis that $\mM$ is simply connected, we construct morphisms of the observable net $\mA_K$ into the global C*-algebra,
describing charges localized in regions of $K$. 
Finally, we illustrate how covariant 1-cocycles shift the localization regions of our localized morphisms, 
confirming their interpretation as charge transporters.

\subsection{Covariant 1-cocycles of the observable net}
\label{b}

We define and start to analyze the covariant cohomology of a given observable net $\mA_K$ fulfilling
the hypothesis of \S \ref{Kn} and \S \ref{rep}. 
The focus is on the  relation between our covariant cohomology and that 
introduced by J.E. Roberts, which provides an equivalent description of DHR superselection sectors
as mentioned in the Introduction.\smallskip

\begin{defn}
\label{b:1}
A \textbf{covariant $1$-cocycle} $X$ is a field 
\[
\mP\times \Sigma_1(K)\ni(\lambda, b)\mapsto X_b(\lambda) \in \mA(|b|)
\]
of unitary operators fulfilling the relation
\begin{equation}
\label{b:2}
X_{\partial_1c}(\sigma\lambda) 
\ = \ 
\alpha^{-1}_\lambda(X_{\lambda\partial_0c}(\sigma)) X_{\partial_2c}(\lambda)
\ \ \ , \ \ \ 
c \in \Sigma_2(K)
\, , \, 
\sigma , \lambda \in \mP
\, ,
\end{equation}
where $\alpha_\lambda := \mathrm{ad}_{U(\lambda)}$.
\end{defn}
The equality (\ref{b:2}) is a covariant extension of the 1-cocycle equation introduced by J.E. Roberts in \cite{Rob90,Rob04}. 
For, if $e \in \mP$ is the identity and $X_{b} := X_{b}(e)$, $b\in\Sigma_1(K)$, then 
\begin{equation}
\label{roberts-cocycle}
X_{\partial_0c} X_{\partial_2c}  \, = \, X_{\partial_1c} \ , \qquad c\in\Sigma_2(K) \, ,
\end{equation}
recovering the defining property of a Roberts 1-cocycle.

Now, it is worth recalling a key property of Roberts 1-cocycles and its consequences. 
Any Roberts 1-cocycle extends from 1-simplices to paths $p=b_n*\cdots *b_1$, by setting $X_p:= X_{b_n}\cdots X_{b_1}$,
and (\ref{roberts-cocycle}) implies that any Roberts 1-cocycle is \emph{homotopy invariant}, that is, 
$X_p=X_q$ for $p$ homotopic to $q$. Since $K$ is simply connected, by \textbf{K7} we find 
\[
X_p=X_q \, , \qquad \forall p,q:a\to o  \, ,
\]
and denoting an arbitrary path from $o$ to $a$ by $p_{o,a}$ we conclude that:
\begin{itemize}
\item $X_{b}=1_{\mH}$ for any degenerate 1-simplex $b$\, i.e. $\partial_1b=\partial_0b$. 
\item $X_{b*\bar b}= X_b \, X_{\bar b}= 1_{\mH}$ and  $X^*_b = X_{\bar b}$, for any 1-simplex $b$, where $\bar b$ is the reverse of $b$. 
\item $\lambda p_{o,a}\simeq p_{\lambda o, \lambda a}$ and  
$X_{\lambda p_{o,a}}= X_{p_{\lambda o,\lambda a}}$ for any $\lambda\in\mP$.  
\end{itemize}
As a consequence of the above properties, for any $o\in \Sigma_0(K)$ we can define 
\begin{equation}
\label{b:5a}
X_o(\lambda):= X_{b}(\lambda) \ , \qquad \partial_0b=o=\partial_1b, \  \lambda\in\mP \ . 
\end{equation}
In fact, if $b$ and $\tilde b$ are degenerate 1-simplices s.t. 
$\partial_0b=\partial_1b=\partial_0\tilde b= \partial_1\tilde b$ and $|b|\subseteq |\tilde b|$, then 
we can define a degenerate 2-simplex $c$ as $\partial_1 c=\partial_2 c=\tilde b$, $\partial_1c=b$ and 
$|c|=|\tilde b|$; by \eqref{b:2} and (\ref{roberts-cocycle}) we have 
\[
 X_b(\lambda) X_{\tilde b}= X_{\partial_0c}(\lambda) X_{\partial_2 c}= X_{\partial_1 c}= X_{\tilde b}(\lambda) \, ,
\]
thus $X_b(\lambda)$ is independent of the support of the degenerate 1-simplex and (\ref{b:5a}) is well-defined. 
We conclude this section by showing an identity that will be extensively used in the sequel. 
\begin{lem}
\label{b:6}
Given a covariant 1-cocycle $X$, for any path $p:o\to \tilde o$ and for any $\lambda\in\mP$ we have 
\[
X_{\lambda p} \,U(\lambda)\, X_o(\lambda)= U(\lambda)\,X_{\tilde o}(\lambda)\,  X_{p} \, .
\]
\end{lem}
\begin{proof}
First assume that $p$ is a 1-simplex $b$. Consider the 2-simplex $c$ defined by 
\[
\partial_0c=b \ \ , \ \ \partial_0c=(\partial_1b,\partial_1b; |b|) \ \ , \ \ 
\partial_1c=b \ \ , \ \ |c|=|b| \ .
\]
Then by \eqref{b:1} we have
\begin{align*}
X_{\lambda b} U(\lambda) X_{\partial_1b}(\lambda) & = X_{\lambda \partial_0c} U(\lambda) X_{\partial_2 c}(\lambda)= U(\lambda)\, \alpha^{-1}_\lambda(X_{\lambda \partial_0c}) X_{\partial_2 c}(\lambda)
=  U(\lambda)\, X_{\partial_1 c}(\lambda) \\
& =  U(\lambda)\, X_{b}(\lambda)
\end{align*}
Considering the 2-simplex  $\tilde{c}$ defined by 
\[
\partial_1\tilde{c}=b \ \ , \ \ \partial_0\tilde c=(\partial_0b,\partial_0b;\partial_0b) \ \ , \ \ \partial_2\tilde c=b \ \ , \ \  |\tilde c|=|b| 
\]
and applying again \eqref{b:1} to the previous identity we get 
\begin{align*}
X_{\lambda b} U(\lambda) X_{\partial_1b}(\lambda) & =  
U(\lambda)\, X_{b}(\lambda)=
U(\lambda)\, X_{\partial_1 \tilde c}(\lambda) =
U(\lambda)\, X_{\partial_0\tilde c}(\lambda)\, X_{\partial_2\tilde c}\\
& =
U(\lambda)\, X_{\partial_0b}(\lambda)\, X_{b} \, .
\end{align*}
Thus if $p:o\to \tilde o$ then the proof follows by iterating the same reasoning to all the 
1-simplices of the path. 
\end{proof}

\subsection{The category of covariant 1-cocycles}

Our final aim is to prove that covariant 1-cocycles are nothing but covariant charge transporters 
of the observable net. To this end we must exhibit, analogously to the DHR analysis, 
a charge structure encoded by a symmetric tensor $\rC^*$-category with conjugates, 
associated with the analogues of DHR and BF endomorphisms. 
In the present section we make a first step towards this direction by showing that 
covariant 1-cocycles form a $\rC^*$-category closed under direct sums and subobjects.

\begin{defn}
\label{def.b:6}
Given covariant 1-cocycles $X$ and $Y$, an \textbf{intertwiner} from $X$ to $Y$
is a field $t:\Sigma_0(K)\ni a\to t_a\in\mA(a)$ satisfying the relation 
\begin{equation}
\label{b:7}
\alpha^{-1}_\lambda (t_{\lambda\partial_0b}) \, X_{b}(\lambda) = Y_{b}(\lambda) \, t_{\partial_1b} \ , \qquad b\in\Sigma_1(K) \ . 
\end{equation}
We denote the set of the intertwiners from $X$ to $Y$ by $(X,Y)$.
\end{defn}
The category of \emph{covariant 1-cocycles} $Z^1_c(\mA_K)$ is the category with objects 
covariant 1-cocycles and arrow the above-defined intertwiners, equipped with 
the composition law  $(X,Y)\times(Z,X) \ni t,s\mapsto (t\cdot s)\in (Z,Y)$,
\[ 
(t\cdot s)_a := t_a\,s_a \ , \qquad a\in\Sigma_0(K) \, .
\]
The identity arrow of a covariant 1-cocycle $X$ is defined by  $(1_X)_a:= 1_{\mH}$ for any $0$-simplex $a$. 
The norm and the involution of $\mB(\mH)$ endow $Z^1_c(\mA_K)$ of a structure of $\rC^*$-category. 
In particular, the adjoint of $t\in (X,Y)$ is defined by 
\[
(t^*)_a:= t^*_a \ , \qquad  a\in\Sigma_0(K) \, . 
\]
To check that $t^*\in (Y,X)$, we observe that by \eqref{b:2} we have 
\[
\alpha^{-1}_\lambda(X_{\lambda \bar b}) X_b(\lambda)= X_{\partial_1b}(\lambda)\iff
 X^*_b(\lambda)= X^*_{\partial_1b}(\lambda) \alpha^{-1}_\lambda(X_{\lambda \bar b}) \, ,
\]
where $\bar b$ is the reverse of $b$. So, using this identity we have
\begin{align*}
\Big(\alpha^{-1}_\lambda(t^*_{\lambda \partial_0b}) Y_{b}(\lambda)\Big)^* & = 
Y^*_{b}(\lambda) \alpha^{-1}_\lambda(t_{\lambda \partial_0b}) = 
Y^*_{\partial_1b}(\lambda) \alpha^{-1}_{\lambda}(Y_{\lambda \bar{b}}) \alpha^{-1}_\lambda(t_{\lambda \partial_0b})\\
& = Y^*_{\partial_1b}(\lambda) \alpha^{-1}_{\lambda}\big(Y_{\lambda \bar{b}}t_{\lambda \partial_0b}\big) =
Y^*_{\partial_1b}(\lambda) \alpha^{-1}_{\lambda}\big(t_{\lambda \partial_1b}X_{\lambda \bar{b}}\big)\\
& = Y^*_{\partial_1b}(\lambda) \alpha^{-1}_{\lambda}\big(t_{\lambda \partial_1b}\big)  \,\alpha^{-1}_{\lambda}\big(X_{\lambda \bar{b}}\big)= t_{\lambda \partial_1b} X^*_{\partial_1b}(\lambda) \,\alpha^{-1}_{\lambda}\big(X_{\lambda \bar{b}}\big)\\
& = t_{\lambda \partial_1b} X^*_b(\lambda) \ .
\end{align*}
Passing to the involution, we get 
$\alpha^{-1}_\lambda(t^*_{\lambda \partial_0b}) Y_{b}(\lambda) = X_b(\lambda)t_{\lambda \partial_1b}$,
showing that $t^*\in(Y,X)$. \smallskip

We recall some notions on $\rC^*$-categories. A \emph{projection} is an arrow $e\in(X,X)$ such that 
$e^*\cdot e=e$; an \emph{isometry} is an arrow $v\in(X,Y)$ such that $v^*\cdot v= 1_X$,
and if $v\cdot v^*=1_Y$ then $v$ is called a \emph{unitary}. Two objects $X,Y$ are said to be \emph{equivalent} 
whenever there is a unitary arrow $u\in(X,Y)$. 
Finally, an object $X$ is said to be \emph{irreducible} if $(X,X)=\bC\cdot 1_X$.
A $\rC^*$-category is closed under \emph{subobjects} if for any projection $e\in(X,X)$ there are an object $Y$ and an isometry 
$v\in(X,Y)$ s.t. $w\cdot w^*=e$, and closed under \emph{(finite) direct sums} if for any pair
of objects $X,Y$ exists an object $Z$ and a pair of isometries $w\in(X,Z)$ and $v\in(Y,Z)$ s.t. 
$w\cdot w^*+v\cdot v^*= 1_Z$.\smallskip

An obvious example of irreducible object of $Z^1_c(\mA_K)$ is the \emph{identity} $I$ defined by 
$I_b(\lambda)_b:= 1_{\mH}$ for any 1-simplex $b$. 
If $t\in(I,I)$ then $t_{\partial_0b}= t_{\partial_0b}I_b= I_bt_{\partial_1b}= t_{\partial_1b}$ for any 1-simplex $b$.
Since $K$ is pathwise connected, this implies that $a\mapsto t_a$ is a constant field, $t_a=\tau$ for any $0$-simplex $a$,
and causality implies $\tau\in \mA(K)\cap \mA(K)''=\mathbb{C}1_{\mH}$. Clearly,
any covariant 1-cocycle $X$ equivalent to $I$ is irreducible.
We conclude this section with the following 
\begin{prop}
$Z^1_c(\mA_K)$ is a $\mathrm{C}^*$-category  closed under direct sums and subobjects. 
\end{prop}
\begin{proof}
We have already seen that $Z^1_c(\mA_K)$ is a $\rC^*$-category. Let us prove that $Z^1_c(\mA_K)$ is closed under subobjects. 
We adapt to the case of covariant 1-cocycles the proof given by Roberts in \cite{Rob90} for 1-cocycles. 
Let $e\in(X,X)$ be a projection. For any $0$-simplex $a$ we choose a $0$-simplex $k(a)$
such that $k(a)\subset a$. By the Borchers property there is an isometry 
$v_a\in \mA(a)$ any  $v_av_a^*=e_{k(a)}$. Now, let $b(a)$ be a 1-simplex with $\partial_1b(a)=k(a)$ and $\partial_0b(a)=|b(a)|=a$.
Define 
%$Z^1_c(\mA_K)$
\[
w_a:= X_{b(a)}v_a \ , \qquad a\in\Sigma_0(K) \ , 
\]
and 
\[
 Y_b(\lambda) := \alpha^{-1}_{\lambda}(w^*_{\lambda\partial_0b})X_{b}(\lambda) w_{\partial_1b} \ , \qquad b\in\Sigma_1(K) \ . 
\]
Observe that $w_a$ is an isometry of $\mA(a)$. In fact $w^*_aw_a=v_a^* X^*_{b(a)}X_{b(a)}v_a=v_a^*v_a=1$ and   
\[
w_aw_a^*=  X_{b(a)}v_a v_a^* X^*_{b(a)}=  X_{b(a)}e_{k(a)}  X^*_{b(a)}\stackrel{\eqref{b:7}}{=} e_aX_{b(a)} X^*_{b(a)} = e_a\ .
\]
Finally 
\begin{align*}
\alpha^{-1}_\lambda (Y_{\lambda\partial_0c}(\sigma))    Y_{\partial_2c}(\lambda) & =
 \alpha^{-1}_\lambda (\alpha^{-1}_{\sigma}(w^*_{\sigma\lambda\partial_{00}c})\, X_{\lambda\partial_0c}(\sigma) w_{\lambda\partial_{10}c})
\alpha^{-1}_{\lambda}(w^*_{\lambda\partial_{02}c})X_{\partial_2 c}(\lambda) w_{\partial_{12}c}\\
& = \alpha^{-1}_\lambda (\alpha^{-1}_{\sigma}(w^*_{\sigma\lambda\partial_{00}c})\, X_{\lambda\partial_0c}(\sigma) w_{\lambda\partial_{10}c})
\alpha^{-1}_{\lambda}(w^*_{\lambda\partial_{10}c})X_{\partial_2 c}(\lambda) w_{\partial_{11}c}\\
& = \alpha^{-1}_{\sigma\lambda}(w^*_{\sigma\lambda\partial_{00}c})\, \alpha^{-1}_{\lambda}(X_{\lambda\partial_0c}(\sigma) e_{\lambda\partial_{10}c}) X_{\partial_2 c}(\lambda) w_{\partial_{11}c}\\
& \stackrel{\eqref{b:7}}{=} \alpha^{-1}_{\sigma\lambda}(w^*_{\sigma\lambda\partial_{00}c})\, \alpha^{-1}_{\lambda}(\alpha^{-1}_{\sigma}(e_{\sigma\lambda\partial_{00c}})X_{\lambda\partial_0c}(\sigma)) X_{\partial_2 c}(\lambda) w_{\partial_{11}c}\\
&  = \alpha^{-1}_{\sigma\lambda}(w^*_{\sigma\lambda\partial_{00}c})\, X_{\partial_1c}(\sigma\lambda) w_{\partial_{11}c}=  \alpha^{-1}_{\sigma\lambda}(w^*_{\sigma\lambda\partial_{01}c})\, X_{\partial_1c}(\sigma\lambda) w_{\partial_{11}c}\\ 
& = Y_{\partial_1c}(\sigma\lambda) \ .
\end{align*}
So $Y$ is a covariant 1-cocycle as well and it is a subobject of $X$ because $w\in (X,Y)$ is an isometry 
s.t. $w\cdot w^*=e$. In a similar way one can prove that $Z^1_c(\mA_K)$ is closed under direct sums.  
\end{proof}

\subsection{Covariant, localized and transportable endomorphisms}
\label{c}

In this section we show that $Z^1_c(\mA_K)$ is equivalent to a category of 
morphisms of $\mA_K$ with values in the global C*-algebra $\mA(K) \subseteq \mB(\mH)$.
The term \emph{morphism} is intended in the sense of \cite{RV11}, 
that is, we have a family $\varrho$ of *-morphisms
\begin{equation}
\label{endAK}
\varrho_a : \mA(a) \to \mA(K)  
\ \ {\mathrm{such \ that}} \ \ 
\varrho_{a'} \restriction \mA(a) = \varrho_a 
\ , \ 
\forall a \subseteq a' \, .
\end{equation}
Our morphisms, that we shall construct starting from $Z^1_c(\mA_K)$, result to be localized, transportable and covariant.
Moreover, in Minkowski spacetime they reduce to the DHR endomorphisms and to the BF morphisms,
picking as sets of indices the ones of double cones and spacelike cones respectively.
Thus we interpret our morphisms as quantum charges,
that result to be transported by covariant cocycles. 
From a mathematical point of view, these morphisms prove to be the key to introducing the tensor product on the category of covariant cocycles (see next section).

\medskip

Let $X$ be a covariant 1-cocycle. Given $o\in K$. For any $a\in K$ we take $\tilde a \perp a$ and define 
\begin{equation}
\label{c:1}
\rho\loc{o}_a(A) \, := \, X_{p_{o,\tilde a}} A X_{p_{\tilde a,o}} \in \mA(K) \ , \qquad A\in\mA(a) \ . 
\end{equation}
This definition is well posed. For, if $\hat a$ another element of $K$ which is causally disjoint 
from $a$, since the causal complement of
$a$ is pathwise connected, there is a path $p_{\hat a,\tilde a}$  contained in 
the causal complement of $a$. So 
\[
X_{p_{o,\tilde a}} A X_{p_{\tilde a,o}} = 
X_{p_{o,\hat a}}X_{p_{\hat a,\tilde a}} A X_{p_{\tilde a,\hat a}} X_{p_{\hat a,\tilde a}}=
X_{p_{o,\hat a}} A  X_{p_{\hat a,\tilde a}}
\]
where we used homotopy invariance and the fact that  $X_{p_{\hat a,\tilde a}}$ commutes with 
$\mA(a)$ because $p_{\hat a,\tilde a}$ is in the causal complement of $a$. 
Moreover, picking $a' \supseteq a$ and $\tilde{a} \perp a'$ we have $\tilde{a} \perp a$, 
implying
$\rho\loc{o}_{a'}(A) =  \rho\loc{o}_a(A)$, 
$A \in \mA(a)$.
Thus the family $\rho\loc{o} := \{\rho\loc{o}_a\}_{a\in K}$ fulfils (\ref{endAK}),
and defines a morphism
\begin{equation}
\label{rholoc}
\rho\loc{o} : \mA_K \to \mA(K) \, .
\end{equation}
We note that $\rho\loc{o}$ can be extended to the algebra generated by a finite number of regions $a_1,\ldots,a_n$
provided that there exists $\tilde a\perp  a_1,a_2, \ldots,   a_n$:
\[
\rho\loc{o}_{a_1}(A_1)\cdots \rho\loc{o}_{a_n}(A_n)= 
X_{p_{o,\tilde a}} A_1 X_{p_{\tilde a,o}}\cdots X_{p_{o,\tilde a}} A_n X_{p_{\tilde a,o}}=
X_{p_{o,\tilde a}} A_1\cdots A_n X_{p_{\tilde a,o}} \, .
\]   
In this case we shall use the notation 
\begin{equation}
\rho\loc{o}_{a_1,a_2,\ldots, a_n}(A_1\cdots
A_n):= X_{p_{o,\tilde a}} A_1\cdots A_n X_{p_{\tilde a,o}} \ , \qquad\tilde a\perp  a_1,a_2, \ldots,   a_n \ . 
\end{equation}
Our aim is now to prove that (\ref{rholoc}) generalize DHR endomorphisms, 
in the sense that they are localized, transportable and, moreover, covariant.
\begin{lem}
\label{c:2}
The following assertions hold:
\begin{itemize}
\item[(i)] if $a\perp o$ then $\rho\loc{o}_a(A)=A$ for any $A\in\mA(a)$.
\item[(ii)] $\rho\loc{o}_{\tilde{o}}(\mA(\tilde o))\subseteq \mA(\tilde o)$ for any $o\subseteq \tilde o$.
\item[(iii)]  
$X_{p_{\hat o,o}}\rho\loc{o}= \rho\loc{\hat o} X_{p_{\hat o,o}}$ for any $\hat o$.
\item[(iv)] $\mathrm{ad}_{X_b(\lambda)}\circ \rho\loc{\partial_1b}_{a} = \alpha^{-1}_{\lambda}\circ \rho\loc{\lambda\partial_0b}_{\lambda a}\circ \alpha_\lambda$.
\end{itemize}
\end{lem}
\begin{proof}
$(i)$ If $a\perp o$ because of $\textbf{K4}$, there is $\tilde a\perp a, o $. Since the causal complement of $a$ is pathwise connected, we may take a path $p_{o,\tilde a}$ laying  in the causal complement of 
$a$. So $\rho\loc{o}_a(A)= X_{p_{o,\tilde a}} A X_{p_{\tilde a,o}} =A$ for any $A\in\mA(a)$. 

$(ii)$ Let $a\perp \tilde o$ and $A\in\mA(a)$. By $(i)$ we have $\rho\loc{o}_a(A)=A$. 
Now, by taking $\tilde a\perp    \tilde{o}, a$  (here we are using again property $\textbf{K4}$ of the poset $K$)
we have that 
\begin{align*}
\rho\loc{o}_{\tilde o}(B)A & = \rho\loc{o}_{\tilde o}(B)\rho\loc{o}_a(A)=
X_{p_{o,\tilde a}} B  X_{p_{\tilde a,o}} X_{p_{o,\tilde a}} A  X_{p_{\tilde a,o}} \\
& =
X_{p_{o,\tilde a}} B A  X_{p_{\tilde a,o}}=
X_{p_{o,\tilde a}} A  X_{p_{\tilde a,o}} X_{p_{o,\tilde a}} B  X_{p_{\tilde a,o}}=
\rho\loc{o}_a(A)\rho\loc{o}_{\tilde o}(B)\\
& = A\rho\loc{o}_{\tilde{o}}(B)
\end{align*}
Since this holds for any $a\perp \tilde o$, the proof follows by relative Haag duality.\smallskip 
 
$(iii)$ Given $A\in\mA(a)$, by homotopy invariance 
\[
X_{p_{\hat o,o}}\rho\loc{o}_a(A) = X_{p_{\hat o,o}}X_{p_{o,\tilde a}} A X_{p_{\tilde a,o}}=
X_{p_{\hat o,\tilde a}} A X_{p_{\tilde a,\hat o}} X_{p_{\hat o,o}}=
\rho\loc{o}_a(A) X_{p_{\hat o,o}} \ . 
\]

$(iv)$ Take $A\in\mA(a)$ and $\tilde a\perp a$. Then 
\begin{align*}
X_b(\lambda)\rho\loc{\partial_1b}_{a}(A) & =
X_b(\lambda) X_{p_{\partial_1b,\tilde a}}A  X_{p_{\tilde a,\partial_1b}}\\
& = 
U^{-1}(\lambda)X_{\lambda b} U(\lambda) X_{\partial_1b}(\lambda) X_{p_{\partial_1b,\tilde a}}A  X_{p_{\tilde a,\partial_1b}}\\
& \stackrel{\eqref{b:6}}{=} 
U^{-1}(\lambda)X_{\lambda b} X_{p_{\lambda \partial_1b,\lambda\tilde a}} U(\lambda) 
X_{\tilde a}(\lambda) A X_{p_{\tilde a,\partial_1b}}\\
&= U^{-1}(\lambda)X_{p_{\lambda \partial_0b,\lambda\tilde a}} U(\lambda) 
 A X_{\tilde a}(\lambda) X_{p_{\tilde a,\partial_1b}}\\
&=  U^{-1}(\lambda)\Big(\rho\loc{\lambda\partial_0b}_{\lambda a} \big(U(\lambda)A U^{-1}(\lambda)\big)\Big)U(\lambda)\, X_b(\lambda) \ , 
\end{align*}
because $X_{\tilde{a}}\in\mA(\tilde a)$ and commutes with $A$.
\end{proof}

\medskip 

\begin{rem}
\label{rem-transport}
In particular, the previous point $(iii)$ shows that the unitaries 
$X_{p_{\hat o,o}}$ intertwine $\rho\loc{o}$ and $\rho\loc{\hat o}$.
This is interpreted as the fact that $X$ transports along the path $p_{\hat o,o}$
the charge defined by $\rho$.
\end{rem}

\medskip 

What remains to be proved is that $\rho:= \{ \rho\loc{o} \}$ is covariant. To this end we define 
\begin{equation}
\label{c:3}
U^\rho_o(\lambda):= X_{p_{o,\lambda o}}U(\lambda) X_o(\lambda) \ , \qquad o\in K \ , \ \lambda\in\mP \ .
\end{equation}
We have the following 
\begin{lem}
\label{c:4}
Let $X$ be a covariant 1-cocycle, and $\rho$, $U^\rho$ be defined as above. Then the following assertions hold:
\begin{itemize}
\item[(i)] $U^\rho_o$ is a unitary representation of $\mP$ for any $o\in K$;
\item[(ii)] $U^\rho_o(\lambda)\rho\loc{o}_a(A)=\rho\loc{o}_{\lambda a}(\alpha_\lambda(A))U^\rho_o(\lambda)$ for any $A\in\mA(a)$ and $\lambda\in\mP$.
\end{itemize}
\end{lem}
\begin{proof}
$(i)$ 
\begin{align*}
 U^\rho_o(\lambda)U^\rho_o(\sigma)& = X_{p_{o,\lambda o}} U(\lambda) X_o(\lambda)X_{p_{o,\sigma o}} U(\sigma) X_o(\sigma)\\
 & \stackrel{Lem.\, \ref{b:6}}{=}
X_{p_{o,\lambda o}} X_{\lambda p_{o,\sigma o}} U(\lambda)X_{\sigma o}(\lambda)U(\sigma) X_o(\sigma)\\
& \stackrel{\eqref{b:2}}{=}
X_{p_{o,\lambda o}} X_{p_{\lambda o,\lambda \sigma o}} U(\lambda)U(\sigma) X_o(\lambda\sigma)
 = X_{p_{o,\lambda \sigma o}} U(\lambda\sigma) X_o(\lambda\sigma)\\
&= U^\rho_o(\lambda\sigma) \ .
\end{align*} 
$(ii)$ Given $A\in\mA(a)$, take $\hat a \perp a$ and by a direct computation we have that 
\begin{align*}
 U^\rho_o(\lambda)\rho\loc{o}_a(A)& = X_{p_{o,\lambda o}} U(\lambda) X_o(\lambda)X_{p_{o,\hat a}} AX^*_{p_{o,\hat a}}
 \\
 &  =
X_{p_{o,\lambda o}}  X_{\lambda p_{o, \hat a} } X_{\lambda  p_{ \hat a,o} }U(\lambda) X_o(\lambda)X_{p_{o,\hat a}} AX^*_{p_{o,\hat a}}\\
& \stackrel{Lem.\,\ref{b:6}}{=}
 X_{p_{o,\lambda o}}X_{\lambda p_{o, \hat a} }U(\lambda) X_{\hat a}(\lambda) X_{p_{ \hat a,o}}X_{p_{o,\hat a}} AX^*_{p_{o,\hat a}}\\ 
& = 
 X_{p_{o,\lambda o}}X_{\lambda p_{o,\hat a} }U(\lambda) A X_{\hat a}(\lambda)X^*_{p_{o,\hat a}} \\
 & = X_{p_{o,\lambda o}} X_{p_{\lambda o, \lambda \hat a}} \alpha_\lambda(A) U(\lambda)\,X_{\hat a}(\lambda) X^*_{p_{o,\hat a}}= X_{p_{o,\lambda \hat a}} \alpha_\lambda(A) U(\lambda)\, X_{\hat a}(\lambda)X^*_{p_{o,\hat a}}\\ 
& = \rho\loc{o}_{\lambda a}(\alpha_\lambda(A)) U^\rho_o(\lambda)   
\end{align*} 
completing the proof.
\end{proof}
The above result motivates the following definition. 
\begin{defn}
A \textbf{localized and transportable} morphism $\rho$ of $\mA_K$ 
is given by a collection $\rho\loc{o} : \mA_K \to \mA(K)$, ${o\in K}$, of morphisms such that  
\begin{itemize}
\item $\rho\loc{o}_a= \mathrm{id}_{\mA_a} $ for any $a\perp o$;
\end{itemize}
and for any pair $o,\tilde o\in K$ there are unitaries  $v^\rho_{o,\hat o}\in\mA$ such that 
\begin{itemize}
\item $v^\rho_{o,\hat o}\rho\loc{\hat o}= \rho\loc{o}\, v^\rho_{o,\hat o}$ and  $v^\rho_{\tilde o,o}v^\rho_{o,\hat o}=v^\rho_{\tilde o,\hat o}$. 
\end{itemize}
A localized and transportable morphism $\rho$ is \textbf{covariant} whenever 
for any $o\in K$ there is a unitary representation $U^\rho_o$ of $\mP$ such that 
\begin{itemize}
\item $\mathrm{ad}_{U^\rho_o(\lambda)}\circ \rho\loc{o}_a= \rho\loc{o}_{\lambda a }\circ \alpha_\lambda$,  for any $\lambda\in \mP$.
\end{itemize}
\end{defn}
Note that the unitaries $\{v^\rho_{o,\tilde o}\}$ are nothing but Roberts 1-cocycles. By the localization of  morphisms and relative Haag duality it follows that 
$v^\rho_{o,\tilde o}\in\mA(a)$ for any $a$ with $o,\tilde o\subseteq a$. Hence,  
Lemmas \ref{c:2} and \ref{c:4} say that any covariant 1-cocycle defines a localized, transportable and covariant morphism.

\begin{rem}
Localizable and transportable morphisms 
$\rho\loc{o} : \mA_K \to \mA(K) \subseteq \mB(\mH)$
are representations of $\mA_K$ on $\mH$ in the sense of \cite{RV11}.
When the set of indices is upward directed, $\rho\loc{o}$ defines a *-endomorphism 
of the global observable algebra $\mA(K)$, see \cite{Rob90,Rob04}. 
This holds in particular in Minkowski spacetime, 
for $K$ given by double cones (so that we recover DHR endomorphisms). 
We also note that we recover the BF-morphisms in the case of spacelike cones.
\end{rem}  
 
\medskip

We now extend at the level of $\rC^*$-categories
the correspondence between covariant 1-cocycles and localized, transportable and covariant morphisms. 
\begin{defn}
\label{c:6}
The set of the intertwiners $t\in(\rho,\gamma)$ between two localized, transportable and covariant morphisms 
$\rho$ and $\gamma$ is the set of fields $t:K\ni o\to t_o\in\mA(o)$ such that   
\begin{itemize}
\item $t_o\rho\loc{o}_a=\gamma\loc{o}_a t_o$  for any $o$ and $a$; 
\item $t_oU^\rho_o(\lambda)= U^\gamma_o(\lambda) t_o$ for any $o$ and $\lambda\in \mP$.
\end{itemize}
\end{defn}
Taking localized, transportable and covariant morphisms of $\mA_K$ as objects and the corresponding set 
of intertwiners, we get a category denoted by $\Delta_c(\mA_K)$. The involution and the norm of $\mB(\mH)$ endow this category  of a structure of a $\rC^*$-category. We call $\Delta_c(\mA_K)$
the category of \emph{localized, transportable and covariant morphisms of $\mA_K$}. 

Next step is to prove that we have a covariant functor from $Z^1_c(\mA_K)$ to $\Delta_c(\mA_K)$. The equations \eqref{c:1} and \eqref{c:3} defines the functor on the objects. To define our functor on the arrows, for any pair $X,Y$ of covariant 1-cocycles we consider
the corresponding morphisms $\rho,\gamma\in \Delta_c(\mA_K)$ and $t\in (X,Y)$.
Then, for any $o\in K$, by definition \eqref{c:1} we have 
\begin{align*}
t_o\, \rho\loc{o}_a(A) & =t_o \, X_{p_{o,\tilde a}} \, A \, X_{p_{\tilde a,o}}=
Y_{p_{o,\tilde a}} \, t_{\tilde a}\, A \, X_{p_{\tilde a,o}} =Y_{p_{o,\tilde a}} \, A\,  t_{\tilde a} \, X_{p_{\tilde a,o}}\\ 
& = Y_{p_{o,\tilde a}} \, A \, X_{p_{\tilde a,o}} \, t_o=\gamma\loc{o}_a(A) \, t_o \  
\end{align*}
and 
\begin{align*}
t_o U^\rho_o(\lambda) & = t_o\, X_{p_{o,\lambda o}}\, U(\lambda)\,  X_o(\lambda)=
Y_{p_{o,\lambda o}}\,  t_{\lambda o}\, U(\lambda)\,  X_o(\lambda)\\
& = 
Y_{p_{o,\lambda o}} \, U(\lambda)\, \alpha^{-1}_\lambda(t_{\lambda o})\, X_o(\lambda)= U^\gamma_o(\lambda)\,  t_o \ . 
\end{align*}
So $t \in (\rho,\gamma)$, and we have a covariant functor from $Z^1_c(\mA_K)$ to $\Delta_c(\mA_K)$. 
\begin{prop}
The $\rC^*$-categories $Z^1_c(\mA_K)$ and $\Delta_c(\mA_K)$ are equivalent.
\end{prop} 
\begin{proof}
We define 
the functor from  $\Delta_c(\mA_K)$ to $Z^1_c(\mA_K)$. Let $\rho\in \Delta_c(\mA_K)$  and denote the corresponding  transporting unitaries and representation of $\mP$ by $v$ and $U$, respectively (we omit  the sup-script $\rho$ to simplify the notation). 
Fix an element  $a\in K$, the \emph{pole}, and define  
\[
X_b(\lambda):=   U^{-1}(\lambda)v^*_{p_{\lambda\partial_0b,a}}U_a(\lambda)v_{p_{a,\partial_1b}}
\]
Clearly while the definition depends on the pole $a$ it is independent of  
the choice of the paths joining the pole $a$ because $v_{a,o}$ are Roberts 1-cocycles. We prove that $X_b(\lambda)\in\mA(|b|)$: for any $o\perp |b|$ and $A\in\mA(o)$ we have 
\begin{align*}
X_b(\lambda)A &=U^{-1}_a(\lambda)v^*_{p_{\lambda\partial_0b,a}}U(\lambda)v_{p_{a,\partial_1b}} \rho\loc{\partial_1b}_o(A)
 = U^{-1}(\lambda)v^*_{p_{\lambda\partial_0b,a}}U_a(\lambda)\rho\loc{a}_o(A)v_{p_{a,\partial_1b}} \\
& = U^{-1}(\lambda)v^*_{p_{\lambda\partial_0b,a}}\rho\loc{a}_{\lambda o}(\alpha_\lambda(A))U_a(\lambda)v_{p_{a,\partial_1b}} \\
& = U^{-1}(\lambda)\rho\loc{\lambda\partial_0b}_{\lambda o}(\alpha_\lambda(A))v^*_{p_{\lambda\partial_0b,a}}U_a(\lambda)v_{p_{a,\partial_1b}} 
 = U^{-1}(\lambda)\alpha_\lambda(A)v^*_{p_{\lambda\partial_0b,a}}U_a(\lambda)v_{p_{a,\partial_1b}} \\
& = A X_b(\lambda)
\end{align*}
and the proof follows by relative Haag duality. For any 2-simplex $c$ we have 
\begin{align*}
\alpha^{-1}_\lambda (X_{\lambda\partial_0c}(\sigma))\,  X_{\partial_2c}(\lambda) & = 
U^{-1}(\lambda)U^{-1}(\sigma) v^*_{p_{\sigma\lambda\partial_{00}c,a}}U_a(\sigma)v_{p_{a,\lambda\partial_{10}c}}
v^*_{p_{\lambda\partial_{02},a}}U_a(\lambda)v_{p_{a,\partial_{12}c}}\\
& =U^{-1}(\sigma\lambda) v^*_{p_{\sigma\lambda\partial_{01}c,a}}U_a(\sigma)v_{p_{a,\lambda\partial_{02}c}}
v^*_{p_{\lambda\partial_{02},a}}U_a(\lambda)v_{p_{a,\partial_{11}c}}\\
& =U^{-1}(\sigma\lambda) v^*_{p_{\sigma\lambda\partial_{01}c,a}}U_a(\sigma)U_a(\lambda)v_{p_{a,\partial_{11}c}}\\
&= U^{-1}(\sigma\lambda) v^*_{p_{\sigma\lambda\partial_{01}c,a}} U_a(\sigma\lambda)v_{p_{a,\partial_{11}c}}
 = X_{\partial_1c}(\sigma\lambda) \ , 
\end{align*} 
and this proves that $X$ is a covariant 1-cocycle. 
Let now $t\in(\rho,\gamma)$ be an intertwiner. We define $\phi(t)_o:= {v^{\rho}_{oa}} \, t_a \, v^\gamma_{ao}$. Note that 
$\phi(t)_o\in\mA(o)$ since for any $\tilde a\perp o$ and $A\in\mA(\tilde a)$ we have
\[
\phi(t)_o \, A=  v^{\rho}_{oa}\, t_a\, v^\gamma_{ao}\, \gamma\loc{o}_{\tilde a}(A)=
v^{\rho}_{oa}\, t_a\,  \gamma\loc{a}_{\tilde a}(A)\, v^\gamma_{ao}= v^{\rho}_{oa}\, \rho\loc{a}_{\tilde a}(A)\, t_a\, v^\gamma_{ao}= 
A\, \phi(t)_o 
\]
and relative Haag duality makes the rest of the proof. Then   
\begin{align*}
\alpha^{-1}_\lambda(\phi(t)_{\lambda \partial_0b})X^\rho_{b}(\lambda) & = \alpha^{-1}_\lambda(\phi(t)_{\lambda \partial_0b})
U^{-1}(\lambda)v^\rho_{\lambda\partial_0b,a}U^\rho_a(\lambda)v^\rho_{a,\partial_1b}\\
& =
U^{-1}(\lambda) \phi(t)_{\lambda \partial_0b}v^\rho_{\lambda\partial_0b,a}U^\rho_a(\lambda)v^\rho_{p_{a,\partial_1b}}
 = U^{-1}(\lambda) v^\rho_{\lambda\partial_0b,a} t_a\, U^\rho_a(\lambda)v^\rho_{a,\partial_1b}\\
& =
U^{-1}(\lambda) v^\gamma_{\lambda\partial_0b,a} U^\gamma_a(\lambda) t_av^\rho_{a,\partial_1b}
 = U^{-1}(\lambda) v^\gamma_{\lambda\partial_0b,a} U^\gamma_a(\lambda) v^\gamma_{a,\partial_1b}\phi(t)_{\partial_1b}\\
& =
X^\gamma_b(\lambda) \phi(t)_{\partial_1b} \ . 
\end{align*}
It is easy to see that this defines a covariant functor. We get an equivalence because by choosing a different pole $\hat a$ we get equivalent objects.
\end{proof}
It is worth observing that the above equivalence maps irreducible covariant 1-cocycles into irreducible covariant morphisms. 
Here, irreducibility is intended in the sense of $\rC^*$-categories 
(whilst the global C*-algebra is assumed to be only factorial).

\subsection{Tensor structure}
\label{tens}
We now focus on $Z^1_c(\mA_K)$, and prove that it is a tensor $\rC^*$-category. 
To this end, given two covariant 1-cocycles $X$ and $Y$ we define 
\begin{equation}
\label{tens:1}
(X\otimes Y)_b(\lambda):= X_{b}(\lambda)\rho{\scriptstyle(\partial_1b)}_{|b|}(Y_{b}(\lambda)) \ , \qquad b\in\Sigma_1(K), \
\lambda\in \mP,
\end{equation}
and for any  $t\in(X,Y)$ and $s\in (Z,L)$ we set
\begin{equation}
\label{tens:2}
(t\otimes s)_a:= t_a\rho{\scriptstyle (a)}_{a} (s_a) \ , \qquad a\in\Sigma_0(K) \ , 
\end{equation} 
where $\rho$ is the morphism defined by $X$.

\begin{prop}
\label{tens:3}
The category $Z^1_c(\mA_K)$, with the tensor product defined by \eqref{tens:1}, \eqref{tens:2}, is a tensor $\rC^*$-category.
\end{prop} 
\begin{proof}
We use the same notation as definitions \eqref{tens:1} and \eqref{tens:2}. 
We note that $X\otimes Y$ is covariant 1-cocycle 
\begin{align*}
(X\otimes Y)_{\partial_1c}(\sigma\lambda) & = 
X_{\partial_1c}(\sigma\lambda)\rho\loc{\partial_{11}c}_{|\partial_1c|}(Y_{\partial_1c}(\sigma\lambda)) 
 = X_{\partial_1c}(\sigma\lambda)\rho\loc{\partial_{12}c}_{|c|}(Y_{\partial_1c}(\sigma\lambda))\\
& =\alpha^{-1}_\lambda\big(X_{\lambda\partial_0c}(\sigma)\big) X_{\partial_2c}(\lambda) \, 
\rho\loc{\partial_{12}c}_{|c|}\Big(\alpha^{-1}_\lambda \big(Y_{\lambda \partial_0c}(\sigma)\big)Y_{\partial_2c}(\lambda)\Big)\\
& =\alpha^{-1}_\lambda\big(X_{\lambda\partial_0c}(\sigma)\big) \, X_{\partial_2c}(\lambda) \, 
\rho\loc{\partial_{12}c}_{|c|}\Big(\alpha^{-1}_\lambda \big(Y_{\lambda \partial_0c}(\sigma)\big)\Big) \rho\loc{\partial_{12}c}_{|c|}\Big(Y_{\partial_2c}(\lambda)\Big)\\
& \stackrel{Lem.\,\ref{c:2}iv}{=} \alpha^{-1}_\lambda\Big(X_{\lambda\partial_0c}(\sigma)\rho\loc{\lambda \partial_{02}c}_{\lambda |c|}\big(Y_{\lambda \partial_0c}(\sigma)\big) \Big) X_{\partial_2c}(\lambda) \, \rho\loc{\partial_{12}c}_{|c|}\big(Y_{\partial_2c}(\lambda)\big)\\
& = \alpha^{-1}_\lambda\Big(X_{\lambda\partial_0c}(\sigma)\rho\loc{\lambda \partial_{10}c}_{\lambda |c|}\big(Y_{\lambda \partial_0c}(\sigma)\big) \Big) X_{\partial_2c}(\lambda) \, \rho\loc{\partial_{12}c}_{|c|}\big(Y_{\partial_2c}(\lambda)\big)\\
& = \alpha^{-1}_\lambda\Big((X\otimes Y)_{\lambda \partial_0c}(\sigma)\Big)\, 
(X\otimes Y)_{\partial_2c}(\lambda) \ . 
\end{align*}
As for intertwiners, 
\begin{align*}
\alpha^{-1}_\lambda ((t\otimes s)_{\lambda\partial_{0b}}) (X\otimes Z)(\lambda)_b& = 
\alpha^{-1}_\lambda ( t_{\lambda\partial_0b}\, \rho\loc{\lambda\partial_0b}_{|\lambda b|} (s_{\lambda \partial_0b})) \, X_{b}(\lambda)\, \rho\loc{\partial_1b}_{|b|}(Z_{b}(\lambda)) \\
& \stackrel{Lem.\,\ref{c:2}iv}{=} 
\alpha^{-1}_\lambda ( t_{\lambda\partial_0b}) \, X_{b}(\lambda) \, \rho\loc{\partial_1b}_{|b|} \big(\alpha^{-1}_\lambda(s_{\lambda \partial_0b})\big) \rho\loc{\partial_1b}_{|b|}\big(Z_{b}(\lambda)\big) \\
& = 
\alpha^{-1}_\lambda ( t_{\lambda\partial_0b}) \, X_{b}(\lambda)\,  \rho\loc{\partial_1b}_{|b|} \big(\alpha^{-1}_\lambda(s_{\lambda \partial_0b})Z_{b}(\lambda)\big) \\
& = 
Z_{b}(\lambda) \, t_{\partial_1b} \, \rho\loc{\partial_1b}_{|b|} \big(L_{b}(\lambda) s_{\partial_1b}\big) \\
& = 
Z_{b}(\lambda) \gamma{\scriptstyle (\partial_1b)}_{|b|} (L_{b}(\lambda)) \,
t_{\partial_1b} \rho{\scriptstyle (\partial_1b)}_{|b|} (s_{\partial_1b})) \\
& = 
(Z\otimes L)_{b}(\lambda)) \,
(t\otimes s)_{\partial_1b} 
\end{align*}
where $\gamma$ is the localized and transportable morphism associated with $Z$. This proves that 
$t\otimes s\in (X\otimes Z,Y\otimes L)$. We omit the proof of the rest of the properties 
of tensor product because follows by standard calculations. 
\end{proof}

\subsection{Permutation symmetry}
\label{sim}

We now prove that the category of covariant 1-cocycles is a symmetric tensor $\rC^*$-category, in other words
we prove the existence of a permutation symmetry. To this end, the first step is to extend the tensor product.
Given two 1-simplex $b_1$ and $b_2$ we define 
\begin{equation}
(X\times Y)_{b_1,b_2}(\lambda) \, := \, X_{b_1}(\lambda)\rho\loc{\partial_1b_1}_{|b_2|}(Y_{b_2}(\lambda)) \, .
\end{equation}
This definition, in the case of Roberts cocycles, admits an extension to pairs of paths.
Given $p,q$ having the same length, $p=b_n*\cdots *b_1$ and $q=d_n*\cdots*d_1$ we define 
\begin{equation}
(X\times Y)_{p,q} \, := \, (X\times Y)_{b_n,d_n}\cdots (X\times Y)_{b_1,d_1} \, .
\end{equation}
We observe that if the paths had different length we may modify the shorter one by inserting 
degenerate 1-simplices. Since a Roberts cocycle equals the identity on degenerate 1-simplices one can easily prove that the above definition is independent 
of where these degenerate 1-simplices are inserted. So we can always assume that in the above definition the paths have the same length.  
We now prove three lemmas necessary to the proof of the existence of the permutation symmetry.
\begin{lem}
\label{sim:1}
Given two covariant 1-cocycles $X$ and $Y$ and two paths $p$ and $q$. For any  pair of paths $\tilde p$, $\tilde q$ homotopy equivalent to $p$ and $q$ respectively, the following relation holds 
\[
(X\times Y)_{p,q}= (X\times Y)_{\tilde p,\tilde q} \ .
\] 
\end{lem}
\begin{proof}
As homotopy equivalent paths are obtained one each other by a finite number of elementary deformations (see Section \ref{Kn}),
it is enough to prove the assertion for an elementary deformation. 
Let $p=b$ and $q=\partial_1c$ is the 1-face of a 2-simplex $c$. 
So $\partial_1c=q\sim \partial_0c*\partial_2c$.  
Then 
\begin{align*}
(X\times Y)_{p,q} & = X_{b}\rho\loc{\partial_1b}_{|c|}(Y_{\partial_1c}) \\
 & = X_{b}\rho\loc{\partial_1b}_{|c|}(Y_{\partial_0c}Y_{\partial_2c})
  =
  X_{b}\rho\loc{\partial_1b}_{|c|}(Y_{\partial_0c})\, \rho\loc{\partial_1b}_{|c|}(Y_{\partial_2c}) \\
  & = X_{b}\rho\loc{\partial_1b}_{|c|}(Y_{\partial_0c})  \, X_{\sigma_0(\partial_1c)}\rho\loc{\partial_1b}_{|c|}(Y_{\partial_2c})\\
  & = (X\times Y)_{b,\partial_0 c}\, (X\times Y)_{\sigma_0(\partial_1b), \partial_2c} \\
  & = (X\times Y)_{b*\sigma(\partial_1b), \partial_0c*\partial_2c}=
   (X\times Y)_{b, \partial_0c*\partial_2c}
\end{align*}
where we used the fact that $X_{\sigma_0(\partial_1c)}=1$. 
\end{proof}
\begin{lem}
\label{sim:2}
Let $p,q$ be two paths such that $\partial_1p\perp\partial_1q$ and $\partial_0p\perp \partial_0q$. Then 
\[
(X\times Y)_{p,q}= (Y\times X)_{q,p} 
\]
\end{lem}
\begin{proof}
First of all we prove that we can find to paths $\tilde p=b_n*\cdots* b_1$ and $\tilde q=d_n*\cdots* d_1$ such that
\[
\partial_i \tilde p=\partial_ip \ , \ \partial_i \tilde q=\partial_iq \ \ for \ i=0,1 \ \ \text{and} \ \ 
b_k\perp d_k \ , \ k=1,\ldots, n \ . 
\]
In fact  by \textbf{K4}, there is $o\subset \partial_1 p$ and $a\in K$ 
such that $a\perp o,\partial_0p$. Then we define  $b_1:=(o,\partial_1p;\partial_1p)$ and as $d_1$ the degenerate 1-simplex $\sigma_0(\partial_1q)$.
Since $|b_1|=\partial_1p\perp\partial_1q=|d_1|$ and the causal complement of the elements of $K$ is pathwise connected \textbf{K6} the proof,  
from this point, is the same as the proof \cite{Rob90}[Lemma 23.6] given 
by Roberts.

As $\tilde p \sim p$ and $\tilde q\sim q$, by \textbf{K7}, we have 
\begin{align*}
(X\times Y)_{p,q} & \stackrel{Lem.\,\ref{sim:1}}{=} (X\times Y)_{\tilde p,\tilde q}
 = 
(X\times Y)_{b_n,d_n}\cdots (X\times Y)_{b_1,d_1} \\
& = 
X_{b_n}\rho\loc{\partial_1 b_n}_{|d_n|}(Y_{d_n})\cdots  X_{b_1}\rho\loc{\partial_1 b_1}_{|d_1|}(Y_{d_1})\\
&=
X_{b_n} Y_{d_n} \cdots  X_{b_1} Y_{d_1} =  Y_{d_n}X_{b_n} \cdots   Y_{d_1} X_{b_1} \\
& =
Y_{d_n}\gamma\loc{\partial_1d_n}_{|b_n|}(X_{b_n}) \cdots   Y_{d_1}\gamma\loc{\partial_1d_1}_{|b_1|}( X_{b_1}) \\
& =  
(Y\times X)_{\tilde q,\tilde p}  \stackrel{Lem.\,\ref{sim:1}}{=}
(Y\times X)_{q,p} 
\end{align*}
where $\gamma$ denotes the morphism associated with $Y$. 
\end{proof}
%%%%%%%%%%%%%%%%%%%%
\begin{lem}
\label{sim:3}
Given $b,\tilde b$ and $b_1,b_2$ with $\partial_0b=\partial_1b_1$ and 
$\partial_0\tilde b=\partial_1b_2$, then  
\[
\alpha^{-1}_\lambda\Big(X_{\lambda b_1}\times Y_{\lambda b_2}\Big)\,  (X_b(\lambda) \times Y_{\tilde b}(\lambda))= 
(X_{b_1}(\lambda)\times Y_{b_2}(\lambda)) \,  (X_{b}\times Y_{\tilde b})
\, .
\] 
\end{lem}  
\begin{proof}
\begin{align*}
\alpha^{-1}_\lambda &\big(X_{\lambda b_1}\times  Y_{\lambda b_2}\big)\,  (X_b(\lambda)\times Y_{\tilde {b}}(\lambda)) =
\alpha^{-1}_\lambda\big(X_{\lambda b_1}\rho\loc{\lambda \partial_1b_1}_{|\lambda b_2|}(Y_{\lambda b_2})\big)\,  
X_{b}(\lambda)\, \rho\loc{\partial_1b}_{|\tilde b|}(Y_{\tilde b}(\lambda))\\
%%%
& = 
\alpha^{-1}_\lambda\big(X_{\lambda b_1}\big)\, \alpha^{-1}_\lambda\big(\rho\loc{\lambda\partial_1b_1}_{|\lambda b_2|}(Y_{\lambda b_2})\big)\, 
X_{b}(\lambda)\, \rho\loc{\partial_1b}_{|\tilde b|}(Y_{\tilde b}(\lambda))\\
%%%
& \stackrel{Lem.\, \ref{c:2} iv)}{=} 
\alpha^{-1}_\lambda\big(X_{\lambda b_1}\big)\, X_{b}(\lambda)\, 
\rho\loc{\partial_1b}_{|b_2|}\big(\alpha^{-1}_\lambda\big(Y_{\lambda b_2}\big)\big)\,  
\rho\loc{\partial_1b}_{|\tilde b|}(Y_{\tilde b}(\lambda)) \\
& =
\alpha^{-1}_\lambda\big(X_{\lambda b_1}\big) \, X_{\partial_0 b}(\lambda) \,X_{b}\,\rho\loc{\partial_1b}_{|b_2|}\big(\alpha^{-1}_\lambda\big(Y_{\lambda b_2}\big)\big)\, 
\rho\loc{\partial_1b}_{|\tilde b|}(Y_{\tilde b}(\lambda))\\
& =
X_{b_1}(\lambda)\,  X_{b}\, \rho\loc{\partial_1b}_{|b_2|}\big(\alpha^{-1}_\lambda\big(Y_{\lambda b_2}\big)\big)\,  
\rho\loc{\partial_1b}_{|\tilde b|}(Y_{\tilde b}(\lambda)) \, . 
\end{align*}
We now observe that $Y_{b_2}(\lambda)Y_{\sigma_0(\partial_1b_2)}^{-1}(\lambda)=\alpha^{-1}_\lambda\big(Y_{\lambda b_2})$ and $Y_{\tilde b}(\lambda)= Y_{\sigma(\partial_0\tilde b)}(\lambda)\, Y_{\tilde b}$.   
Inserting these identities in the last expression gives
\begin{align*}
\alpha^{-1}_\lambda &\big(X_{\lambda b_1}\times  Y_{\lambda b_2}\big)\,  (X_b(\lambda)\times Y_{\tilde {b}}(\lambda)) =\\
%
%X_{b_1}(\lambda)\,  X_{b}\, \rho\loc{\partial_1b}_{|b_2|}\big(\alpha^{-1}_\lambda\big(Y_{\lambda b_2}\big)\big)\,  
%\rho\loc{\partial_1b}_{|\tilde b|}(Y_{\tilde b}(\lambda))  \\
& = X_{b_1}(\lambda)\,  X_{b}\,  \rho\loc{\partial_1b}_{|b_2|}\big(Y_{b_2}(\lambda)Y_{\sigma_0(\partial_1b_2)}^{-1}(\lambda)\big)\, \rho\loc{\partial_1b}_{|\tilde b|}(Y_{\sigma(\partial_0\tilde b)}(\lambda)\, Y_{\tilde b}) \\
& = X_{b_1}(\lambda)\,  X_{b}\,  \rho\loc{\partial_1b}_{|b_2|}\big(Y_{b_2}(\lambda)\big)\,
\rho\loc{\partial_1b}_{\partial_1b_2}\big(Y_{\sigma_0(\partial_1b_2)}^{-1}(\lambda)\big)\, 
\rho\loc{\partial_1b}_{\partial_0\tilde b}(Y_{\sigma(\partial_0\tilde b)}(\lambda)) \, 
\rho\loc{\partial_1b}_{|\tilde b|}(Y_{\tilde b})\\
& = X_{b_1}(\lambda)\,  X_{b}\,  \rho\loc{\partial_1b}_{|b_2|}\big(Y_{b_2}(\lambda)\big)\,
\rho\loc{\partial_1b}_{|\tilde b|}(Y_{\tilde b})\\
& = X_{b_1}(\lambda)\, \rho\loc{\partial_0b}_{|b_2|}\big(Y_{b_2}(\lambda)\big) \,X_{b}\, 
\rho\loc{\partial_1b}_{|\tilde b|}(Y_{\tilde b})\\
& = X_{b_1}(\lambda)\times Y_{b_2}(\lambda)\,  X_{b}\times Y_{\tilde b} \ .
\end{align*}
where we used the identity
\[
\rho\loc{\partial_1b}_{\partial_1b_2}\big(Y_{\sigma_0(\partial_1b_2)}^{-1}(\lambda)\big)\, 
\rho\loc{\partial_1b}_{\partial_0\tilde b}(Y_{\sigma(\partial_0\tilde b)}(\lambda)) =1 \, 
\]
which holds because $\partial_0\tilde b=\partial_1b_2$.
\end{proof}
Note that the above identity is equivalent to 
\begin{equation}
\label{sim:3a}
(X_b(\lambda) \times Y_{\tilde b}(\lambda))\, (X_{b}\times Y_{\tilde b})^* = 
\alpha^{-1}_\lambda\Big(X_{\lambda b_1}\times Y_{\lambda b_2}\Big)^*\,(X_{b_1}(\lambda)\times Y_{b_2}(\lambda)) \ .
\end{equation}

We are ready to prove the existence of a permutation symmetry. Given $X,Y$ two covariant 1-cocycles define 
\begin{equation}
\label{sim:4}
\epsilon(X,Y)_a:= (Y\times X)^*_{q,p} \,  (X\times Y)_{p,q} \ , \qquad a\in\Sigma_0(K) \ ,
\end{equation}
where $p,q$ are two paths with $\partial_1p=\partial_1q=a$ and $\partial_0p\perp\partial_1q$. \smallskip

We first prove that this definition is independent of the choice  of paths. Let $p_1$ and $q_1$ be two paths satisfying the same properties as 
$p$ and $q$. Clearly $p\sim  p*\overline{p}_1*p_1$ and $q\sim  q*\overline{q}_1*q_1$. 
By Lemma \ref{sim:1} we have 
\[
(X\times Y)_{p,q}= (X\times Y)_{p*\overline{p}_1*p_1,q*\overline{q}_1*q_1}=
(X\times Y)_{p*\overline{p}_1,q*\overline{q}_1}\, (X\times Y)_{p_1,q_1}
\]
Since  $(X\times Y)_{p*\overline{p}_1,q*\overline{q}_1}= (Y\times X)_{q*\overline{q}_1, p*\overline{p}_1}$
because of Lemma \ref{sim:2}, we get  
\[
(Y\times X)^*_{q,p} (X\times Y)_{p,q}= (Y\times X)^*_{q_1,p_1} (X\times Y)_{p_1,q_1} \ ,
\]
showing the independence of the choice of paths. 
\begin{prop}
\label{sim:5}
Given $X,Y$ covariant 1-cocycles, the operators $\epsilon(X,Y)$ defined by \eqref{sim:4} yield a permutation symmetry for $Z^1_c(\mA_K)$.
\end{prop}
\begin{proof}
Given a 1-simplex $b$, according to the definition of intertwiner \ref{b:7},
and since $\eps$ does not depend on the choice of paths, for each $\lambda\in \mP$, 
we take two paths of the form $\lambda p, \lambda q$, where 
$p,q$ are paths paths satisfying the properties of the definition \eqref{sim:4} with respect to $\partial_0b$.  It is evident that $\lambda p,\lambda q$ satisfy 
the properties of the definition \eqref{sim:4} with respect to 
$\lambda \partial_0 b$. Then applying the Lemma \ref{sim:3} once, we have  
\begin{align*}
\alpha^{-1}_\lambda & (\epsilon(X,Y)_{\lambda \partial_0 b}) (X\otimes Y)_b(\lambda)  = 
\alpha^{-1}_\lambda ((Y\times X)^*_{\lambda q,\lambda p} \cdot (X\times Y)_{\lambda p,\lambda q}) (X_b(\lambda)\times Y_b(\lambda))\\
& = \alpha^{-1}_\lambda ((Y\times X)^*_{\lambda q,\lambda p}) \cdot 
\alpha^{-1}_\lambda ((X\times Y)_{\lambda p,\lambda q}) (X_b(\lambda)\times Y_b(\lambda))\\
& = \alpha^{-1}_\lambda ((Y\times X)^*_{\lambda q,\lambda p}) \cdot 
\alpha^{-1}_\lambda ((X\times Y)_{\lambda b_n,\lambda d_n}) \cdots \alpha^{-1}_\lambda ((X\times Y)_{\lambda b_1,\lambda d_1})(X_b(\lambda)\times Y_b(\lambda))\\
& \stackrel{Lem. \ref{sim:3} }=
\alpha^{-1}_\lambda ((Y\times X)^*_{\lambda q,\lambda p}) \cdot 
\alpha^{-1}_\lambda ((X\times Y)_{\lambda b_n,\lambda d_n}) \cdots 
(X_{b_1}(\lambda)\times Y_{d_1}(\lambda)) X_b\times Y_b \ . 
\end{align*}
Applying Lemma \ref{sim:3}  iteratively, gives 
\begin{align*}
\alpha^{-1}_\lambda  &(\epsilon(X,Y)_{\lambda \partial_0 b}) (X\otimes Y)_b(\lambda) 
=  \\
&= \alpha^{-1}_\lambda ((Y\times X)^*_{\lambda q,\lambda p}) \cdot
 X_{b_n}(\lambda)\times Y_{d_n}(\lambda)\cdot  
(X_{b_{n-1}}\times Y_{d_{n-1}})\cdots  X_b\times Y_b\\
& =
\alpha^{-1}_\lambda ((Y\times X)^*_{\lambda q,\lambda p}) \cdot 
 X_{\partial_0 b_n}(\lambda)\times Y_{\partial_0 d_n}(\lambda)\cdot  
   (X\times Y)_{p*b, q*b}\\
   & =
 \alpha^{-1}_\lambda ((Y\times X)^*_{\lambda q,\lambda p})   \cdot 
 Y_{\partial_0 d_n}(\lambda)\times X_{\partial_0 b_n}(\lambda)\cdot  
   (X \times Y)_{p*b, q*b} \ ,
\end{align*}
where  being $\partial_0 d_n \perp  \partial_0 b_n$ we used 
$X_{\partial_0 b_n}(\lambda)\times Y_{\partial_0 d_n}(\lambda)=
Y_{\partial_0 d_n}(\lambda)\times X_{\partial_0 b_n}(\lambda)$.
Now, by iteratively applying to the left-hand product the equation \eqref{sim:3a}, which is nothing but the Lemma \eqref{sim:3}, we arrive at
\begin{align*}
\alpha^{-1}_\lambda  (\epsilon(X,Y)_{\lambda \partial_0 b}) (X\otimes Y)_b(\lambda) 
&=   
 Y_{b}(\lambda)\times X_{b}(\lambda)\cdot (Y\times X)^*_{q*b,p*b}   \cdot  
   (X \times Y)_{p*b, q*b}\\
& = (Y\otimes X)_{b}(\lambda)\, \eps(X,Y)_{\partial_1b} \ , 
\end{align*}
where  
\[
\eps(X,Y)_{\partial_1b}=(Y\times X)^*_{q*b,p*b} \cdot (X \times Y)_{p*b, q*b}
\]
because $\epsilon$ does not depends on the choice of paths.  The proof that 
\[
 \epsilon(\tilde X,\tilde Y)_a \,(t\otimes s)_a = (s\otimes t)_a\, \epsilon (X,Y)_a \ , \qquad a\in\Sigma_0(K) 
\]  
for any $t\in(X,\tilde X)$ and $s\in(Y,\tilde Y)$ follows by a similar reasoning.  
We omit the rest of the properties of permutation symmetry since they follow by standard calculations.
\end{proof}

\subsection{Statistics and conjugation}
\label{con}
In this section we select the subcategory of covariant 1-cocycles having finite statistics and prove 
that any object of this category has conjugates in the sense of symmetric tensor $\mathrm{C}^*$-categories. 
In our context, the \emph{conjugate} of a covariant 1-cocycle $X$ 
is a covariant 1-cocycle $\bar X$ for which there exists a pair of arrows  $r\in (I,\bar X\otimes X)$ and 
$\bar r\in (I,X\otimes \bar X)$ that satisfy the \emph{conjugate equations}  
\begin{equation}
\label{ceq}
\bar r^*\otimes 1_X\cdot 1_X\otimes r= 1_X \ \ , \ \ r^*\otimes 1_{\bar X}\cdot 1_{\bar{X}}\otimes \bar{r}= 1_{\bar X} \, .
\end{equation}
The key result is that simple covariant 1-cocycles, i.e.\, those obeying Fermi or Bose statistics, 
have conjugates. This will allows us to identify the subcategory of objects having finite statistics 
and to prove the existence of conjugates.\smallskip

\begin{defn}
A covariant 1-cocycle $X$ is said to be \textbf{simple} whenever 
\begin{equation}
\label{con:0}
\epsilon(X,X)= \chi(X) \cdot 1_{X\otimes X} \ \ , \ \ \chi(X)\in\{1,-1\} \ . 
\end{equation}
\end{defn}
We shall see at the end of this section that simple covariant 1-cocycles obey either  Bose or Fermi statistics,
depending on whether the value of $\chi(X)$ is $1$ or $-1$.\smallskip

We now draw on a consequence of this relation. According to the definition of permutation symmetry \eqref{sim:4}, for any pair $o\perp \tilde o$  
we take a path $p_{o,\tilde o}:o\to\tilde o$ and as $q$ the degenerate 1-simplex $\sigma_0(o)=(o,o;o)$ and observe that 
\[
 \chi(X) = \rho\loc{o}_{|b_n|}(X^*_{b_n})\cdots \rho\loc{o}_{|b_1|}(X^*_{b_1}) \,  X_{p_{o,\tilde o}}  \iff 
  \rho\loc{o}_{|b_1|}(X_{b_1})\cdots \rho\loc{o}_{|b_n|}(X_{b_n}) = \chi(X)\, X_{p_{o,\tilde o}} \, .
\]
Changing the role, in this relation, of $o$ and $\tilde o$ and passing to the adjoint we arrive 
\begin{equation}
\label{con:1}
\rho\loc{o}_{|b_1|}(X_{b_1})\cdots \rho\loc{o}_{|b_n|}(X_{b_n}) = \chi(X) X_{p_{o,\tilde o}}=
\rho\loc{\tilde o}_{|b_1|}(X_{b_1})\cdots \rho\loc{\tilde o}_{|b_n|}(X_{b_n})
\end{equation}
Since by $\textbf{K4}$ there is $a\perp o,\tilde o$,  we can assume that $p_{o,\tilde o}$ is  in the causal complement of $a$. So \eqref{con:1} reduces to  
\begin{equation}
\label{con:2}
%\rho\loc{o}_{|p_{o,\tilde o}|}(X_{p_{o,\tilde o}})= \rho\loc{o}_{|b_1|}(X_{b_1})\cdots \rho\loc{o}_{|b_n|}(X_{b_n})
%= \chi(X) \, X_{p_{o,\tilde o}} = \rho\loc{\tilde o}_{|p_{o,\tilde o}|}(X_{p_{o,\tilde o}})
\rho\loc{o}_{|p_{o,\tilde o}|}(X_{p_{o,\tilde o}})= \chi(X) X_{p_{o,\tilde o}} = \rho\loc{\tilde o}_{|p_{o,\tilde o}|}(X_{p_{o,\tilde o}})
\end{equation}
Now our aim is to prove that any simple covariant 1-cocycle admits a conjugate covariant 1-cocycle. 
Here, because of the fact that morphisms associated to 1-cocycles are defined locally, we have 
to follows a different route with respect to the usual one (see \cite{DHR71,DHR74,Ruz05}).
We first prove that the Roberts 1-cocycle defined by a covariant cocycle  
has a conjugate. Then we prove that the morphism associated with this conjugate 1-cocycle 
inverts the morphism defined by the cocycle $X$.
This, finally, will allow us to define the conjugate of a covariant 1-cocycle.

\begin{lem}
Let $X$ be a simple covariant 1-cocycle. For any 1-simplex $b$  the definition 
\begin{equation}
\bar{X}_b:= X_{p_{a,\partial_0b}} X_{p_{\partial_1b, a}} \ , \qquad  a\in\Sigma_0(K) \ , \ a\perp |b|
\end{equation}
is independent of the choice of $a$ and defines a Roberts 1-cocycle.
\end{lem}
\begin{proof}
So let us start by showing the independence of the choice of $a$.  
So take $\tilde a\perp b$. We first assume that $\tilde a \perp a$. So we may take a path 
$p_{a,\tilde a}=d_m*\cdots *d_1$ in the causal complement of $b$ and
\begin{align*}
 X_{p_{a,\partial_0b}} & X_{p_{\partial_1b, a}} =  X_{p_{a,\tilde a}}  X_{p_{\tilde a,\partial_0b}} X_{p_{\partial_1b, \tilde a}} X_{p_{\tilde a,a}} \\
 & \stackrel{\eqref{con:1}}{=}
\rho\loc{\tilde a}_{|d_n|}(X_{d_n})\cdots \rho\loc{\tilde a}_{|d_1|}(X_{d_1}) X_{p_{\tilde a,\partial_0b}} X_{p_{\partial_1b, \tilde a}} 
\rho\loc{\tilde a}_{|d_1|}(X_{\bar d_1})\cdots \rho\loc{\tilde a}_{|d_m|}(X_{\bar d_m})\\
& = X_{p_{\tilde a,\partial_0b}} \rho\loc{\partial_0b}_{|d_m|}(X_{d_m})\cdots 
\rho\loc{\partial_0 b}_{|d_1|}(X_{d_1})\cdot 
\rho\loc{\partial_1b}_{|d_1|}(X_{\bar d_1})\cdots \rho\loc{\partial_1 b}_{|d_m|}(X_{\bar d_m})  X_{p_{\partial_1b, \tilde a}} \\
& = X_{p_{\tilde a,\partial_0b}}  X_{d_m} \cdots X_{d_1}\cdot X_{\bar d_1}\cdots X_{\bar d_m}  X_{p_{\partial_1b, \tilde a}} \\
& =  X_{p_{\tilde a,\partial_0b}}\, X_{p_{\partial_1b, \tilde a}} \ .
\end{align*}
Now let $\tilde a\perp |b|$ but $\tilde a\not\perp a$. Since the causal complement 
of $|b|$ is pathwise connected there is a path $p_{a,\tilde a}=b_n*\cdots*b_1:a\to \tilde a$ 
in the causal complement of $|b|$.
Clearly $\partial_1b_1=a$ and $\partial_0b_n=\tilde a$. Since the support of $b_1$ is spacelike 
separated from the support of $b$ there is, by \textbf{K4},  $o\perp (|b_1|\cup |b|)$. Note in particular 
that $o\perp a, \partial_0b_1$. So applying the above argument first 
with respect to $o$ and $a$ and then with respect to $o$ and $\partial_0b_1$ arrive to 
\[
X_{p_{a,\partial_0b}} X_{p_{\partial_1b, a}} =  X_{p_{o,\partial_0b}}\, X_{p_{\partial_1b, o}}=
X_{p_{\partial_0b_1,\partial_0b}}\, X_{p_{\partial_1b, \partial_0b_1}}
\]
So by iterating this idea to all the 1-simplices of the path we arrive to 
\[
X_{p_{a,\partial_0b}} X_{p_{\partial_1b, a}} =  
X_{p_{\tilde a ,\partial_0b}}\, X_{p_{\partial_1b, \tilde a}}
\]
We now prove that $\bar X$ is a Roberts 1-cocycle. First of all 
observe that $\bar{X}_b\in\mA_{|b|}$. In fact for any $\tilde a \perp |b|$
we may take $a\perp \tilde a, |b|$ and the paths $p_{\partial_1b,a},   p_{a,\partial_0b}$ in the causal complement of $\tilde a$; if $A\in \mA_{\tilde a}$ then 
\[
\bar{X}_b A = X_{p_{a,\partial_0b}} X_{p_{\partial_1b, a}} A=
A X_{p_{a,\partial_0b}} X_{p_{\partial_1b, a}}= A \bar{X}_b
\]
and the proof follows by relative Haag duality. Secondly
given a 2-symplex $c$ take $a\perp |c|$ we have 
\begin{align*}
\bar{X}_{\partial_{0}c}\bar{X}_{\partial_{2}c} & = X_{p_{a,\partial_{00}c}} X_{p_{\partial_{10}c},a} \, X_{p_{a,\partial_{02}c}} \, X_{p_{\partial_{12}c,a}}=
X_{p_{a,\partial_{01}c}} X_{p_{\partial_{02}c},a} \, X_{p_{a,\partial_{02}c}} \, X_{p_{\partial_{11}c,a}}\\
& = X_{p_{a,\partial_{01}c}} \, X_{p_{\partial_{11}c,a}} = \bar{X}_{\partial_1c}
\end{align*}
\end{proof}
Since $\bar X$ is a Roberts 1-cocycle, it defines by \eqref{c:1} 
a morphism $\bar{\rho}$ of the net $\mA_K$. 
\begin{lem}
\label{con:2a}
Let $X$ be a simple covariant 1-cocycle. Let $\bar\rho$  be the morphism of the net 
associated with $\bar X$ by \eqref{c:1}. Then 
\[
\bar\rho\loc{o}_{\tilde o}\circ \rho\loc{o}_{\tilde o} = \rho\loc{o}_{\tilde o}\circ \bar\rho\loc{o}_{\tilde o} = id_{\mA_{\tilde o}} \ , \qquad o\subseteq \tilde o
\]
\end{lem} 
\begin{proof}
We need a preliminary result. 
Let $p:o\to\tilde o$ be a path which lays in the causal complement of $a$,  then 
\begin{equation}
\label{con:3}
\bar{X}_{p}=  X_{p_{a,\tilde o}} X_{p_{o,a}}
\end{equation}
In fact if $p=b_n*\cdots *b_1$, using the independence of the choice of $a$ proved in the previous Lemma  
we have
\begin{align*}
\bar{X}_{p} & = \bar{X}_{b_n}\cdots \bar{X}_{b_1} =
X_{p_{a,\partial_0 b_n}} X_{p_{\partial_1b_n,a}}X_{p_{a,\partial_0 b_{n-1}}} X_{p_{\partial_1b_{n-1},a}} \cdots \
X_{p_{a,\partial_0 b_1}} X_{p_{\partial_1b_{1},a}} \\
& = X_{p_{a,\partial_0 b_n}} X_{p_{\partial_1b_{n-1},a}} \cdots \
X_{p_{a,\partial_0 b_1}} X_{p_{\partial_1b_{1},a}} \\
& = X_{p_{a,\partial_0 b_n}}  X_{p_{\partial_1b_{1},a}}=  X_{p_{a,\tilde o}}  X_{p_{o,a}}
\end{align*}
where we used the fact that $ \partial_1 b_i=\partial_0 b_{i-1}$ and the fact that 
Roberts 1-cocycle equals 1 when evaluated on loops. \smallskip 

Now we take $\tilde o\in K$ with $o\subseteq \tilde o$ and $a\perp \tilde o$. Then 
we consider $\hat a \perp a\cup \tilde o$ and a path $p_{\hat a,a}$ in the causal complement of $\tilde o$.
Then for any $A\in\mA_{\tilde o}$, since $\rho\loc{o}_{\tilde o}(\mA(\tilde o))\subseteq \mA(\tilde o)$ (Lem. \ref{c:2}$(ii)$) we have 
\begin{align*}
\bar{\rho}\loc{o}_{\tilde o} (\rho\loc{o}_{\tilde o})( A ) & =
\bar{X}_{p_{o,a}}(\rho\loc{o}_{\tilde o}( A )) \bar{X}_{p_{a,o}}
\stackrel{\eqref{con:3}}{=}  X_{p_{\hat a, o}} X_{p_{a,\hat a}} \rho\loc{o}_{\tilde o}(A) X_{p_{\hat a,a}} X_{p_{o,\hat a}} \\
& =  X_{p_{\hat a,o}} \rho\loc{o}_{\tilde o}(A) X_{p_{o,\hat a}}\stackrel{Lem. \ref{c:2}(iii)}{=} 
\rho\loc{\hat a}_{\tilde o}(A)=A
\end{align*}
where we have used the fact that  $p_{a,\hat a}$ is in the causal complement of 
$\tilde o$ and $\rho\loc{o}_{\tilde o}(A)\in\mA(\tilde o)$. 
Conversely  
\begin{align*}
 \rho\loc{o}_{\tilde o}(\bar{\rho}\loc{o}_{\tilde o}( A )) & =
 X_{p_{o,a}} \bar{\rho}\loc{o}_{\tilde o}( A ) X_{p_{a,o}}
=X_{p_{\hat a,a}} X_{p_{a,\hat  a}}X_{p_{o,a}} \bar{\rho}\loc{o}_{\tilde o}( A ) X_{p_{a,o}} X_{p_{\hat a,a}} 
X_{p_{a,\hat a}}\\
& \stackrel{\eqref{con:3}}{=} X_{p_{\hat a,a}} \bar{X}_{p_{\hat a,o}} \bar{\rho}\loc{o}_{\tilde o}( A ) \bar{X}_{p_{o,\hat a}} X_{p_{a,\hat a}}
 = X_{p_{\hat a,a}}  \bar{\rho}\loc{\hat a}_{\tilde o}( A ) X_{p_{a,\hat a}}\\
& = X_{p_{\hat a,a}}  A  X_{p_{a,\hat a}}=A
\end{align*}
\end{proof}
We now are ready to prove the existence of the conjugate of  simple and covariant 1-cocycles. 
\begin{thm}
\label{con:4}
Given a simple covariant 1-cocycle $X$ then 
\begin{equation}
\label{con:5}
\bar{X}_b(\lambda):= \bar{\rho}\loc{\partial_1b}_{|b|}(X^*_b(\lambda))\, \qquad b\in\Sigma_1(K) \, , \ \lambda\in\mP \, ,  
\end{equation}  
is the conjugate covariant 1-cocycle of $X$. In particular $X$ and $\bar X$ are irreducible. 
\end{thm}
\begin{proof}
By Lemma \ref{c:2}$ii)$  we have that $\bar X_b(\lambda)\in \mA(|b|)$ for any 
$\lambda\in \mP$. For any 2-simplex $c$ and any $\lambda,\sigma\in\mP$ we have 
\begin{align*}
\alpha^{-1}_{\lambda}\big(\bar{X}_{\lambda\partial_0c}(\sigma)\big)\bar{X}_{\partial_{2c}}(\lambda) & = 
\alpha^{-1}_{\lambda}\big(\bar{\rho}\loc{\lambda\partial_{10}c}_{|c|}(X^*_{\lambda\partial_0c}(\sigma))\big)\bar{\rho}\loc{\partial_{12}c}_{|c|}\big(X^*_{\partial_{2c}}(\lambda)\big)\\
& \stackrel{Lem.\, \ref{con:2a}}{=} \bar{\rho}\loc{\partial_{12c}}_{|c|}\Big\{ \rho\loc{\partial_{12c}}_{|c|}\Big(\alpha^{-1}_{\lambda}\big(\bar{\rho}\loc{\lambda\partial_{10}c}_{|c|}(X^*_{\lambda\partial_0c}(\sigma)\big)\Big)
 X^*_{\partial_{2c}}(\lambda)\Big\}\\
& = \bar{\rho}\loc{\partial_{12c}}_{|c|}\Big\{ X_{\partial_{2}c}(\lambda) \rho\loc{\partial_{12}c}_{|c|}\Big(\alpha^{-1}_{\lambda}\big(\bar{\rho}\loc{\lambda\partial_{10}c}_{|c|}(X_{\lambda\partial_0c}(\sigma)\big)\Big)
\Big\}^*\\
& \stackrel{Lem.\, \ref{c:2}iv)}{=} \bar{\rho}\loc{\partial_{12}c}_{|c|}\Big\{\alpha^{-1}_\lambda\Big(\rho\loc{\lambda\partial_{02}c}_{|c|}\big(\bar{\rho}\loc{\lambda\partial_{10}c}_{|c|}(X_{\lambda\partial_0c}(\sigma)\big)\Big)X_{\partial_{2}c}(\lambda) 
\Big\}^*\\
& = \bar{\rho}\loc{\partial_{12c}}_{|c|}\Big\{X^*_{\partial_{2c}}(\lambda) \alpha^{-1}_\lambda\Big(\rho\loc{\lambda\partial_{02c}}_{|c|}\big(\bar{\rho}\loc{\lambda\partial_{10}c}_{|c|}(X^*_{\lambda\partial_0c}(\sigma)\big)\Big)
\Big\}\\
& \stackrel{\ref{con:2a}}{=}\bar{\rho}\loc{\partial_{12c}}_{|c|}\Big\{X^*_{\partial_{2c}}(\lambda) \alpha^{-1}_\lambda\Big(X^*_{\lambda\partial_0c}(\sigma)\Big)\Big\}\\
& = \bar{\rho}\loc{\partial_{11c}}_{|c|}\Big\{X^*_{\partial_{1c}}(\lambda)\Big\} = 
\bar{X}_{\lambda\partial_0c}(\sigma\lambda) \ .
\end{align*}
So $\bar X$ is a covariant 1-cocycle. Concerning the conjugate equations, given $b$ take $a\perp |b|$ and 
\[
(X\otimes \bar{X})_b(\lambda)  = X_{b}(\lambda)\rho\loc{\partial_1b}_{|b|}(\bar{X}_b(\lambda))
 = X_{b}(\lambda)\rho\loc{\partial_1b}_{|b|}(\bar{\rho}\loc{\partial_1b})_{|b|}(X^*_b(\lambda)) =
 X_{b}(\lambda)X^*_b(\lambda) = 1_{\mH} 
\]
and similarly $(\bar{X}\otimes X)_b(\lambda)=1_{\mH}$. So both $X\otimes \bar X$ and $\bar X\otimes X$ are equal to the 
identity cocycle $I$ and the conjugate equations \eqref{ceq} are verified by taking $\bar r$ and $r$ equal to 
$1_I=1_{\mH}$. Finally in order to prove that $X$ is irreducible, let 
$t\in(X,X)$. By 
\eqref{con:0} and by the definition of tensor product we have $1_X\otimes t\cdot \epsilon(X,X)=\epsilon(X,X) \cdot t\otimes 1_X$ and this implies that $\rho\loc{o}_o(t_o)=t_o$ for any $o\in K$. So by Lemma \ref{con:2a} we also have 
$\bar{\rho}\loc{o}_o(t_o)= (\bar{\rho}\loc{o}_o\circ \rho\loc{o}_o)(t_o)= t_o$ for any $o\in K$. Hence for any 1-simplex $b$ we have 
\[
t_{\partial_0b}=\bar\rho\loc{\partial_0b}_{\partial_0b}(t_{\partial_0b})=(1_{\bar X}\otimes t)_{\partial_0b} (\bar X\otimes X)_b= (\bar X\otimes X)_b (1_{\bar X}\otimes t)_{\partial_1b}= \bar\rho\loc{\partial_1b}_{\partial_1b}(t_{\partial_1b})=t_{\partial_1b} 
\]
because,  as observed above, $(\bar X\otimes X)_b=1_{\mH}$ for any 1-simplex $b$.
Since $K$ is pathwise connected $t$ is a constant field and causality implies that $t$ is a multiple of the identity,
completing the proof.
\end{proof}

We now  introduce the notion of objects with finite statistics. To this end we recall 
that a \emph{left inverse of $X$} is a linear map 
$\phi_{Z,Y}:(X\otimes Z,X\otimes Y)\to (Z,Y)$ satisfying the relations
\begin{itemize}
\item $\phi_{Z\otimes \tilde X, Y\otimes \tilde X}(t\otimes 1_{\tilde X})  =
\phi_{Z, Y}(t)\otimes 1_{\tilde X}$; 
\item $\phi_{Z',Y'}(1_X\otimes s\cdot  t\cdot  1_X\otimes r )  = s \cdot \phi_{Z,Y}(t)\cdot r$  , 
\end{itemize}
for any  $t\in(X\otimes Z, X\otimes Y)$, $s\in (Z,Z')$ and $r\in (Y,Y')$. A left inverse of $X$ 
is said to be \emph{positive} whenever, for any object $Y$, $\phi_{Y,Y}$ sends positive
elements of $(X\otimes Y,X\otimes Y)$ into positive elements of $(Y,Y)$; \emph{normalized}
whenever  $\phi_{I,I}(1_X)= 1_I$ where $I$ is the identity object of the category. 
A positive normalized left inverse $\phi$ of $X$ is said to be \emph{standard} whenever 
$(\phi_{X,X}(\epsilon(X,X)))^2= c\cdot 1_X$ with  $c >0$. 
\begin{defn}
A covariant 1-cocycle $X$ is said to have \textbf{finite statistics} if it admits a standard left inverse. 
The full subcategory of $Z^1_c(\mA_K)$ of objects having finite statistics is denoted by 
$Z^1_{\mathrm{c,f}}(\mA_{K})$.
\end{defn}
We note that any simple covariant 1-cocycle $X$ has finite statistics. 
Since $\bar X\otimes X=X\otimes \bar X=I$ and $r=\bar r= 1_I$, defining   
\[
\phi_{Z,Y}(t):=  1_{\bar X}\otimes t  \ , \qquad t\in(X\otimes Z,X\otimes Y) \, ,
\]
we get a positive, normalized left inverse which is, by definition of simple objects, standard. 
In particular, we note that $\phi_{Z,Y}(t)_a=\bar\rho\loc{a}_a(t_a)$ for any  $a\in\Sigma_0(K)$.
\begin{prop}
The category $Z^1_{\mathrm{c,f}}(\mA_{K})$ is a symmetric tensor $\rC^*$-category closed under tensor products, subobjects and direct sums. Any object of $Z^1_{\mathrm{c,f}}(\mA_{K})$
is a finite direct sum of irreducible objects.
\end{prop}
\begin{proof}
The proof follows from the properties of standard left inverses, see \cite{DHR71, Rob90}.
\end{proof} 
If $X$ is an irreducible covariant 1-cocycle with finite statistics, following \cite{DHR71}, and $\phi$ is a left inverse of $X$,  one has that 
\[
\phi_{X,X}(\epsilon(X,X))= \frac{\chi(X)}{d(X)} \cdot  1_X 
\]
where $\chi(X)\in\{-1,1\}$  and  $d(X)\in\mathbb{N}$, called, respectively, the statistical \emph{phase} and 
\emph{dimension} are invariant of the equivalence class of $X$. These invariants means that
$X$ has a para-statics of order $d(X)$ of Bose or Fermi type depending on whether 
$\chi(X)$ is $1$ or $-1$. Note in particular that simple covariant 1-cocycle follows ordinary Bose or Fermi statistics.

Having shown that  simple covariant 1-cocycles have finite statistics and have conjugates Lemma \ref{con},
we now are ready to give  the main result of the present paper.
\begin{thm}
The category of covariant 1-cocycles with finite statistics $Z^1_{\mathrm{c,f}}(\mA_{K})$ has conjugates.
\end{thm}
\begin{proof}
The fact that $Z^1_{\mathrm{c,f}}(\mA_{K})$ has conjugates derives from the fact that simple objects have conjugates
(see the appendix of \cite{Ruz05}). 
\end{proof}

 \section{Conclusions and outlooks}
\label{sec.concl}

In the current study, we have introduced physically motivated properties that define appropriate families (sets of indices) 
of spacetime regions where quantum charges are expected to be localized.
Then, given such a set of indices and an observable net fulfilling factoriality and relative Haag duality, 
we constructed a covariant superselection structure wherein charges are localized within the aforementioned regions. 
This achievement was made possible by employing a novel approach based on covariant 1-cocycles. 
We emphasize that the definitions and proofs given in this paper do not rely on the symmetry being a group, 
but rather on its semigroup structure. 
Therefore, our approach remains valid even when the symmetry of the charge localization regions is a semigroup.

\medskip 

Our method allows to recover the sectors of the DHR and BF analysis and, in curved spacetimes, 
the sectors of DHR type encoded by the cohomology of the observable net. 
Yet, in the case where the ambient spacetime is the light cone and the set of indices is the one of hypercones,
it remains an open problem to understand the relation between our superselection structure
and the one defined by Buchholz and Roberts for charges of electromagnetic type.
We believe that this question is of interest because,
were the two superselection structures inequivalent, we would be led to conclude that 
picking the set of indices does not uniquely determine the superselection structure with that localization.
In other words, besides the choice of the localization regions, a further input is needed 
to discriminate the superselection structure of interest among those having the same localization.

\medskip 

Another open point is relative to our hypothesis that the ambient spacetime is simply connected. 
This simplified our exposition and it made possible to rule out Aharonov-Bohm external potentials
to which at the present stage we are not interested. 
Yet, we wish to discuss our superselection sectors in full generality, 
with the final goal to arrive to the reconstruction of the field net without limitations on the 
topology of the spacetime. 
In this sense the localized morphisms constructed in \S \ref{c} should play an important role, 
analogous to the one played by DHR endomorphisms in Minkowski spacetime. 
But, in a spacetime with non-trivial fundamental group, 
one has to take into account the fact that localized morphisms exhibit 
a non-trivial parallel transport, depending on the homotopy class of the path along which they are translated. 

\medskip 

Finally, we would like to mention that, thanks to the level of abstraction that we adopted, 
our methods can be applied to low dimensional spacetimes, with the natural modifications 
arising from the fact that permutation symmetry could be replaced by a braiding in specific situations. 
We may therefore apply our suitably modified construction in such scenarios. 
At this purpose, we note that models studied in the language adopted in the present paper are available 
\cite{Cio09}, thus they constitute candidates for further examples of superselection structures
that can be obtained with our method.

\ 

\noindent \textbf{Acknowledgements.}
The authors sincerely thank D. Buchholz, for insightful comments and observations which improved the manuscript, and S. Carpi, G. Morsella, R. Longo for fruitful discussions.
The authors gratefully acknowledge GNAMPA-INdAM. GR  acknowledges the Excellence Project
2023-2027 MatMod@TOV awarded to the Department of Mathematics, University of Rome
Tor Vergata. EV and FC gratefully acknowledge the kind hospitality of the Department of Mathematics, University of Rome ``Tor Vergata” in the final stage of the present work.

%
%%%%%%%%%%%%%%%%%%%%%%%%%%%%%%%%%%%%%%%%%%%%%%%%%%%%%%%%%%%%%%%%%%%%%%%%%%%%%%%%%%%%%%%%
%


{\small
\begin{thebibliography}{99}


\bibitem{BR09}
Brunetti, R., Ruzzi, G.: Quantum Charges and Spacetime Topology: The Emergence of New Superselection Sectors. 
Commun. Math. Phys. \textbf{287}, 523–563 (2009) 
%https://doi.org/10.1007/s00220-008-0671-6

\bibitem{Buc82}
Buchholz, D.:
The physical state space of quantum electrodynamics.
Comm. Math. Phys. \textbf{85}, 49-71 (1982)
%\href{http://projecteuclid.org/euclid.cmp/1103921339}{[open access]}.

\bibitem{BCRV16}
Buchholz, D., Ciolli, F., Ruzzi, G., Vasselli, E.:
The Universal C*-algebra of the Electromagnetic Field.
Lett. Math. Phys.  \textbf{106}, 269-285 (2016)

\bibitem{BCRV17}
Buchholz, D., Ciolli, F., Ruzzi, G., Vasselli, E.:
The universal C*-algebra of the electromagnetic field II. 
Topological charges and spacelike linear fields.
Lett. Math. Phys.  \textbf{107}, 201–222 (2017)

%\bibitem{BCRV19a}
%Buchholz, D., Ciolli, F., Ruzzi, G., Vasselli, E.:
%Linking numbers in local quantum field theory.
%Lett. Math. Phys.,  \textbf{109}, 829–842 (2019)

\bibitem{BCRV19}
Buchholz, D., Ciolli, F., Ruzzi, G., Vasselli, E.:
On string-localized potentials and gauge fields.
Lett. Math. Phys.  \textbf{109}, 2601–2610 (2019)

\bibitem{BCRV22}
Buchholz, D., Ciolli, F., Ruzzi, G., Vasselli, E.:
The universal algebra of the electromagnetic field III. 
Static charges and emergence of gauge fields.
Lett. Math. Phys.  \textbf{112}, 27 (2022)

%\bibitem{BCRV23}
%Buchholz, D., Ciolli, F., Ruzzi, G., Vasselli, E.:
%Gauss's law, the manifestations of gauge fields, and their impact on local observables.
%Accepted for publication \url{https://arxiv.org/abs/2212.11009}


\bibitem{BCRV23}
Buchholz, D., Ciolli, F., Ruzzi, G., Vasselli, E.:
Gauss's law, the manifestations of gauge fields, and their impact on local observables.
In 
Cinto, A., Michelangeli, A. (eds)
``Trails in Modern Theoretical and Mathematical Physics: 
A Volume in Tribute to Giovanni Morchio" 
Springer ISBN-13 978-3031449871 (in press).
\url{https://arxiv.org/abs/2212.11009}


\bibitem{BF82}
Buchholz, D., Fredenhagen, K.: 
Locality and the structure of particle states. 
Commun.Math. Phys. \textbf{84}, 1–54 (1982) 
%https://doi.org/10.1007/BF01208370

\bibitem{BR14}
Buchholz, D., Roberts, J.E.: 
New Light on Infrared Problems: Sectors, Statistics, Symmetries and Spectrum. 
Commun. Math. Phys. \textbf{330}, 935–972 (2014) 
%https://doi.org/10.1007/s00220-014-2004-2

%\bibitem{Buc13}
%D. Buchholz:
%New Light on Infrared Problems: Sectors, Statistics, Spectrum and All That. \\
%\url{http://arxiv.org/abs/1301.2516}.

\bibitem{BDMRS96}
Buchholz, D., Doplicher, S., Morchio G., Roberts, J.E., Strocchi, F. (1997).
A model for charges of electromagnetic type. 
In Doplicher, S., Longo, R., Roberts, J.E., Zsido, L. (eds)
``Operator algebras and quantum field theory" (Rome, 1996), pp. 647-660. 
Internat. Press, Cambridge, MA.

\bibitem{BDMRS07}
Buchholz, D., Doplicher, S., Morchio, G., Roberts, J.E., Strocchi, F. (2007)
Asymptotic Abelianness and Braided Tensor C*-Categories.
In de Monvel, A.B., Buchholz, D., Iagolnitzer, D., Moschella, U. (eds.)
``Rigorous Quantum Field Theory, A Festschrift for Jacques Bros" pp. 49-64. 
Progress in Mathematics, 251. Birkhäuser, Basel.

\bibitem{Cam07}
Camassa, P.:
Relative Haag Duality for the Free Field in Fock Representation. 
Ann. H. Poincar\'e \textbf{8},  1433-1459 (2007)

\bibitem{Cio09}
Ciolli, F.:
Massless scalar free field in 1+1 dimensions I: Weyl algebras, products and superselection sectors.
Rev. Math. Phys. \textbf{21} 735-780 (2009) \\ 
Ciolli, F.: Massless scalar free Field in 1+1 dimensions, II: Net Cohomology and Completeness of Superselection Sectors.
Preliminary version \url{https://arxiv.org/abs/0811.4673}

\bibitem{CRV12}
Ciolli, F., Ruzzi, G., Vasselli, E.:
Causal posets, loops and the construction of nets of local algebras for QFT.
Adv. Theor. Math. Phys. \textbf{16}, 645–692 (2012)

\bibitem{CRV15}
Ciolli, F., Ruzzi, G., Vasselli, E.:
QED representation for the net of causal loops.
Rev. Math. Phys. \textbf{27}, 1-37, 1550012 (2015) 

\bibitem{DRV20}
Dappiaggi, C., Ruzzi, G., Vasselli,E.:
Aharonov-Bohm superselection sectors. 
Lett. Math. Phys. \textbf{110}, 3243-3278 (2020)

\bibitem{DHR69a}
Doplicher S., Haag R., Roberts, J.E.:
Fields, observables and gauge transformations. I.
Comm. Math. Phys. \textbf{13},  1-23 (1969) 
%\url{http://projecteuclid.org/euclid.cmp/1103841481}

\bibitem{DHR69b}
Doplicher S., Haag R., Roberts, J.E.:
Fields, observables and gauge transformations. II.
Comm. Math. Phys. \textbf{15}, 173-200 (1969) 
%\url{http://projecteuclid.org/euclid.cmp/1103841943}

\bibitem{DHR71}
Doplicher, S., Haag, R.,  Roberts, J.E.: 
Local observables and particle statistics I. 
Commun.Math. Phys. \textbf{23}, 199–230 (1971) 
%\url{https://doi.org/10.1007/BF01877742}

\bibitem{DHR74}
Doplicher, S., Haag, R.,  Roberts, J.E.:
Local observables and particle statistics II.
Commun. Math. Phys. \textbf{35}, 49–85 (1974)
%https://doi.org/10.1007/BF01646454

\bibitem{DR90}
Doplicher, S., Roberts, J.E.:
Why there is a field algebra with a compact gauge group describing the superselection structure in particle physics.
Commun. Math. Phys. \textbf{131}, 51–107 (1990) 
%https://doi.org/10.1007/BF02097680

\bibitem{FRS89}
Fredenhagen, K., Rehren, K.-H., Schroer, B.:  
Superselection sectors with braid group statistics and exchange algebras. I. General theory
Commun. Math. Phys.\textbf{125}, 201-226 (1989)

\bibitem{FRS92}
Fredenhagen, K., Rehren, K.-H., Schroer, B.: 
Superselection sectors with braid group statistics and exchange algebras. II. Geometric aspects and conformal covariance.
Rev. Math. Phys. 113–157 (1992)

\bibitem{FMS79A}
Froehlich, J., Morchio, G., Strocchi, F.:
Charged sectors and scattering states in quantum electrodynamics.
Ann. Phys. \textbf{119}, 241-284 (1979) 

\bibitem{GLRV01}
Guido, D.,  Longo, R.,  Roberts, J.E., Verch, R.:
Charged sectors, spin and statistics in quantum field theory on curved spacetimes. 
Rev. Math. Phys. \textbf{13}, 125–198 (2001)

\bibitem{HK64}
Haag, R., Kastler, D.:
An Algebraic Approach to Quantum Field Theory.
J. Math. Phys. \textbf{5}, 848-862 (1964) 
%\href{http://dx.doi.org/10.1063/1.1704187}{[article]}.

\bibitem{Mun09}
Mund, J.: The spin-statistics theorem for anyons and plektons in d = 2+1. 
Commun. Math. Phys. \textbf{286},
1159-1180 (2009)

\bibitem{Naa11}
Naaijkens, P.: 
On the Extension of Stringlike Localised Sectors in 2+1 Dimensions. 
Commun. Math. Phys. \textbf{303}, 385-420 (2011)


\bibitem{Rob90}
Roberts, J. E. (1990). Lectures on algebraic quantum field theory. 
In Kastler, D. (ed.) ``The algebraic theory of superselection sectors" (Palermo, 1989), pp. 1-112.
World Sci. Publ., River Edge, NJ.

\bibitem{Rob04}
Roberts, J.E. (2004). More Lectures on Algebraic Quantum Field Theory. 
In: Doplicher, S., Longo, R. (eds) ``Noncommutative Geometry" (Martinafranca, 2000), pp. 263–342 
Lecture Notes in Mathematics, vol 1831. Springer, Berlin, Heidelberg.  

\bibitem{RV11}
Ruzzi, G., Vasselli, E.:
A New Light on Nets of C*-Algebras and Their Representations.
Commun. Math. Phys. \textbf{312}, 655–694 (2012) 
%url://doi.org/10.1007/s00220-012-1490-3

\bibitem{Ruz05}
Ruzzi, G.:  
Homotopy of posets, net-cohomology and superselection sectors in globally hyperbolic spacetimes.
Rev. Math. Phys. \textbf{17}, 1021-1070 (2005)

\bibitem{Strocchi}
F. Strocchi:
An introduction to non-perturbative foundations of Quantum Field Theory.
International Series of Monographs on Physics 158.
Oxford Science Publications, 2013.

\bibitem{Vas15}
E. Vasselli:
Presheaves of superselection structures in curved spacetimes. 
Comm. Math. Phys. \textbf{335}, 895-933  (2015)

\bibitem{Vas19}
E. Vasselli:
Background potentials and superselections sectors. 
J. Geom. Phys. \textbf{139}, 139-148 (2019)


\end{thebibliography}
}
\end{document}